%% file: 2023-BBQ.tex
\documentclass[12pt]{article}

\usepackage[margin = 1in]{geometry}
\usepackage{setspace}
\usepackage{amsmath,amssymb,amsthm}
\usepackage{natbib}
\usepackage{graphicx}
\usepackage{booktabs}
\usepackage{tikz}
\usepackage{bm}
\usepackage{comment}
\usetikzlibrary{arrows.meta, shapes.geometric, calc, positioning}

\newcommand{\bX}{\bm{X}}
\newcommand{\bY}{\bm{Y}}

\newcommand{\Data}{\mathcal D}
\newcommand{\diag}{\operatorname{diag}}
\newcommand{\Dirichlet}{\operatorname{Dirichlet}}
\newcommand{\E}{\mathbb E}

\newcommand{\Gam}{\operatorname{Gam}}
\newcommand{\iid}{\stackrel{\text{iid}}{\sim}}
\newcommand{\Leaves}{\mathcal L}

\newcommand{\Normal}{\operatorname{Normal}}

\newcommand{\Pit}{\widetilde{\Pi}}
\newcommand{\pit}{\widetilde{\pi}}

\newcommand{\Poisson}{\operatorname{Poisson}}
\newcommand{\Reals}{\mathbb{R}}
\newcommand{\sM}{\mathcal M}
\newcommand{\sS}{\mathcal S}
\newcommand{\Tree}{\mathcal T}
\newcommand{\Treef}{\operatorname{Tree}}
\newcommand{\Uniform}{\operatorname{Uniform}}
\newcommand{\Var}{\operatorname{Var}}

\theoremstyle{theorem}
\newtheorem{theorem}{Theorem}
\newtheorem{lemma}{Lemma}

\theoremstyle{remark}
\newtheorem{remark}{Remark}

\usepackage[colorlinks, citecolor = blue]{hyperref}

\begin{document}

\title{Bayesian Nonparametric Quasi Likelihood}

\author{Antonio R. Linero\thanks{Department of Statistics and Data Sciences,
    University of Texas at Austin, \texttt{antonio.linero@austin.utexas.edu}}}
\date{}

\maketitle

\begin{abstract}
  A recent trend in Bayesian research has been revisiting generalizations of the
  likelihood that enable Bayesian inference without requiring the specification
  of a model for the data generating mechanism. This paper focuses on a Bayesian
  nonparametric extension of Wedderburn's quasi-likelihood, using Bayesian
  additive regression trees to model the mean function. Here, the analyst posits
  only a structural relationship between the mean and variance of the outcome.
  We show that this approach provides a unified, computationally efficient,
  framework for extending Bayesian decision tree ensembles to many new settings,
  including simplex-valued and heavily heteroskedastic data. We also introduce
  Bayesian strategies for inferring the dispersion parameter of the
  quasi-likelihood, a task which is complicated by the fact that the
  quasi-likelihood itself does not contain information about this parameter;
  despite these challenges, we are able to inject updates for the dispersion
  parameter into a Markov chain Monte Carlo inference scheme in a way that, in
  the parametric setting, leads to a Bernstein-von Mises result for the
  stationary distribution of the resulting Markov chain. We illustrate the
  utility of our approach on a variety of both synthetic and non-synthetic
  datasets.
  
\end{abstract}

\doublespacing

\section{Introduction}

We consider the problem of performing Bayesian inference when the analyst would
like to avoid specifying a model for the data generating mechanism. To perform
fully-Bayesian inference, one typically begins by specifying (i) a prior
distribution $\pi(\theta)$ for the parameters $\theta$ and (ii) a likelihood
$L(\theta \mid \Data)$ for the data $\Data$. While a common complaint is that
inferences might be sensitive to the choice of $\pi(\theta)$, a more
consequential problem is that it is rarely the case that one has a strong
justification for $L(\theta \mid \Data)$ either. To make inferences robust to
misspecification of $L(\theta \mid \Data)$, many recent works have aimed to
bypass the specification of the likelihood entirely, and instead model the
information contained in the data using, for example, moment equations
\citep{chib2018bayesian,yin2009bayesian,schennach2005bayesian} or loss functions
\citep{lyddon2019general,bissiri2016general}.

In the context of regression, one possible solution is to specify a
\emph{quasi-likelihood} $Q(\theta)$
\citep{nelder1987extended,mcculaugh1989generalized}, which only requires
specifying a mean function $\E_\theta(Y_i \mid X_i = x) = \mu_\theta(x)$ and a
mean-variance relation that relates $\Var_\theta(Y_i \mid X_i = x)$ to
$\mu_\theta(x)$. A \emph{quasi-posterior} for $\theta$ can then be derived by
replacing $L(\theta \mid \Data)$ with $Q(\theta)$ in Bayes rule, i.e.,
$\pit(\theta) \propto \pi(\theta) \, Q(\theta)$. Despite this, quasi-likelihood
methods have seen limited attention in the Bayesian literature; notable
exceptions include \citet{agnoletto2023bayesian}, \citet{lin2006quasi} and
several works of Ventura \citep{greco2008robust,ventura2010default,
  ventura2016pseudo}. This may be because most common problems solved by
quasi-likelihood --- specifically, accounting for overdispersion for count and
binomial type data --- have alternate Bayesian solutions that instead specify a
full data generating mechanism.

This manuscript revisits the quasi-likelihood, primarily for nonparametric
Bayesian inference, and argues that it offers a convenient, robust, and
computationally efficient alternative to the use of full probabilistic models
for routine data analysis problems. We make the following contributions:

\begin{enumerate}
\item We develop nonparametric models for the mean function $\mu_\theta(x)$
  using Bayesian additive regression tree (BART) models, and show that these
  approaches are computationally efficient and perform very well on both real
  and synthetic data. This effectively reduces our modeling assumptions to a
  single assumption about the mean-variance relation. Computations proceed via
  simple modifications to the Bayesian backfitting algorithms of
  \citet{chipman2010bart} and \citet{murray2021log}. As a byproduct, we obtain
  extensions of BART to both simplex-valued outcome data and structured
  variance models that were previously unavailable.
\item A limitation of using $Q(\theta)$ for inference is that it cannot be used
  directly for inferences about the \emph{dispersion parameter} $\phi$ of the
  model \citep{davidian1988note,nelder1987extended}, a parameter that is
  analogous to the error variance in linear regression. We address this problem
  for Bayesian methods by introducing three different schemes for updating
  $\phi$ during Markov chain Monte Carlo (MCMC) based inference, and further
  extend our approaches to allow for estimation of the mean-variance relation as
  well. 
\item A theoretical issue is that the best-performing methods for updating
  $\phi$ during MCMC correspond to updates from ``full conditionals'' that are
  not compatible with the quasi-likelihood \citep[see][for a discussion of
  compatibility]{gelman2018characterizing}. We therefore give theoretical
  justification to our approach in the parametric setting by establishing a
  Bernstein-von Mises result for our quasi-posterior, with the quasi-posterior
  of the regression parameters being asympototically normal with a
  correctly-calibrated variance.
\end{enumerate}

We illustrate the use of our Bayesian nonparametric quasi-likelihood methods on
synthetic data and data from the Medical Expenditure Panel Survey (MEPS), which
is an ongoing panel survey of the use of the healthcare system in the United
States. On synthetic data, we find that BART-based quasi-likelihood methods
perform extremely well and are no more computationally intensive to use than the
original methods of \citet{chipman2010bart}. On the MEPS dataset, we confirm
existing results in the literature that the variance in medical expenditures
scales roughly as $V(\mu) = \mu^\kappa$ for some $\kappa \in (1,2)$, and further
use our techniques to understand which variables are most material to
forecasting an individual's medical expenditures.

\subsection{Notation and Literature Review}

We consider data $\Data = (\bY, \bX)$ where $\bY = (Y_1, \ldots, Y_N)^\top$
denotes the outcome of interest and $\bX \in \Reals^{N \times P} = (X_1, \ldots,
X_N)^\top$ is a design matrix of covariates, with each $X_i \in \Reals^P$.
Associated to each outcome $Y_i$ is a positive weight $\omega_i$; this allows
for us to accommodate binomial-type data with differing counts per observation,
heteroskedastic data with known weights, and survey sampling weights. We define
$\E_\theta(Y_i \mid X_i = x) = \mu_\theta(x)$ and assume that the mean and
variance are related through the \emph{mean-variance relation} $\Var_\theta(Y_i
\mid X_i, \omega) = \phi\,V_\eta\{\mu_\theta(X_i)\} / \omega_i$; for example,
for Poisson data we have $V_\eta(\mu) = \mu$ while for gamma distributed data we
have $V_\eta(\mu) = \mu^2$. Here, $\theta$ denotes a collection of all of the
model parameters; it consists of a regression function $r(\cdot)$ that
determines $\mu_\theta(\cdot)$, a dispersion parameter $\phi > 0$ and, possibly,
a parameter $\eta$ for the \emph{variance function} $V_\eta$.

Quasi-likelihood (QL) methods were introduced by \citet{weddernburn1974quasi}
shortly after the introduction of generalized linear models (GLMs) by
\citet{nelder1972generalized}. Rather than defining an exponential dispersion
family as done with GLMs, we use the specified mean and variance functions
to define a \emph{quasi-likelihood function}
\begin{align*}
  q_\theta(y \mid x, \omega) =
  \exp\left\{ -\frac{\omega D\{y, \mu_\theta(x)\}}{2\phi} \right\}
  \qquad \text{and} \qquad
  D(y, \mu) = 2 \int_\mu ^ y \frac{y - t}{V_\eta(t)} \ dt,
\end{align*}
where $D(y, \mu)$ is referred to as the \emph{quasi-deviance}. In the special
case where $[Y_i \mid X_i, \omega_i, \theta]$ is modeled with an exponential
dispersion family, it can be shown that the quasi-likelihood coincides with the
associated density up-to a term that does not depend on $\mu_\theta(x)$. This
definition is motivated by the fact that, when the mean and variance functions
are correctly specified, the quasi-likelihood mimics the usual score and
information equations
\begin{align}
  \label{eq:moments}
  \begin{split}
    \E_{\theta}\left\{\frac{\partial}{\partial \mu}
    \log q_\theta(Y_i \mid X_i, \omega_i) \mid X_i\right\} &= 0 
                                                           \qquad \text{and} \qquad \\
    \Var_{\theta}\left\{\frac{\partial}{\partial \mu}
    \log q_\theta(Y_i \mid X_i, \omega_i) \mid X_i\right\}
    &= - \E_\theta\left\{ \frac{\partial^2}{\partial \mu^2} 
       \log q_\theta(Y_i \mid X_i, \omega_i) \mid X_i \right\}.
  \end{split}
\end{align}
Because these equations are effectively the same as those used to justify
likelihood-based inference for GLMs, it is unsurprising that the
quasi-likelihood $Q(\theta) = \prod_i q_\theta(Y_i \mid X_i, \omega)$ behaves in
many ways like a genuine likelihood function even when $q_\theta(y \mid x,
\omega)$ is not a valid density; for example, the quasi-likelihood can be used
to perform likelihood ratio tests and construct confidence intervals as though
it corresponded to a genuine likelihood \citep[][Chapter
9]{mcculaugh1989generalized}.

Despite the appeal of quasi-likelihood methods, Bayesian variants are
surprisingly underdeveloped. One possible strategy for Bayesian inference is to
define a \emph{quasi-posterior} based on data as
$\pi(\theta \mid \Data) = Q(\theta) \, \pi(\theta) / \int Q(\theta) \,
\pi(\theta) \ d\theta$. In addition to having an interpretation in terms of
quasi-likelihood, following \citet{bissiri2016general} the quasi-posterior can
be justified as giving a coherent update of $\theta$ based on the prior and the
loss function $D(y, \mu)$. When $\phi$ and $\eta$ are fixed by the prior, this
approach works quite well; for example, \citet{agnoletto2023bayesian} prove a
Bernstein-von Mises result when the mean and variance are correctly specified
and $\phi$ is known.

\subsection{Dispersion and Variance Function Estimation}

A complication of the use of the quasi-posterior is that the quasi-likelihood
cannot be used to infer $\phi$ \citep[][Section 9.6]{mcculaugh1989generalized}.
The reason for this is that quasi-likelihoods and genuine likelihoods only agree
up to a constant $c(y, \phi)$, and while this constant is immaterial to
estimating $\mu_\theta(x)$ it is very important for estimating $\phi$. While
there have been some attempts at specifying the form of $c(y,\phi)$ for
quasi-likelihood methods to allow estimation of $\phi$
\citep{nelder1987extended}, these solutions may not lead to consistent
estimation of $\phi$ \citep{davidian1988note}.

The most common solution, and the default in most software, is to replace $\phi$
with a moment-based estimator
\begin{align}
  \label{eqn:moment}
  \widehat\phi = \frac{1}{N - P} \sum_{i = 1}^N
  \frac{\omega_i (Y_i - \widehat \mu_i)^2}
  {V(\widehat \mu_i)},
\end{align}
where the $\widehat \mu_i$'s are the maximum likelihood estimators of the
$\mu_\theta(X_i)$'s and $P$ is the number of parameters in the
$\mu_\theta(\cdot)$ model; importantly, the $\widehat \mu_i$'s do not depend on
$\phi$. This strategy is effective in parametric non-Bayesian settings, but it
is not clear how to apply it in Bayesian settings where it is unclear what to
replace $N - P$ with and estimation of the $\mu_\theta(X_i)$'s is no longer
decoupled with estimation of $\phi$.

To aaddress this, we introduce three methods --- extended quasi-posteriors
(EQPs), pseudo-likelihood posteriors (PLPs), and Bayesian bootstrapped
quasi-likelihood (BBQ) --- for estimating $(\phi, \eta)$ during MCMC-based
inference. As there are subtle issues with the compatibility of the PLP and BBQ
methods with the quasi-posterior, we also give some theoretical justification
for these methods.

\subsection{Outline}

In Section~\ref{sec:bayesian-quasi-likelihood}, we describe further our Bayesian
nonparametric quasi-likelihood approach, discuss its extension to multivariate
quasi-likelihood (MQL), and review the use of Bayesian additive regression trees
(BART) to flexibly model $\mu_\theta(x)$. In
Section~\ref{sec:inference-for-the-dispersion-parameter} we introduce our
approaches to estimating the dispersion parameter and variance function, and in
Section~\ref{sec:theoretical-properties} we give Frequentist justification to
these approaches. In Section~\ref{sec:illustrations} we provide illustrations of
the use of Bayesian nonparametric quasi-likelihood on both the MEPS dataset and
synthetic data. We close in Section~\ref{sec:discussion} with a discussion.

\section{Bayesian Nonparametric Quasi-Likelihood}
\label{sec:bayesian-quasi-likelihood}

Let $\theta = (r, \phi, \eta)$ denote the parameter vector under the moment
assumptions $g\{\mu_\theta(x)\} = r(x)$ and $\Var_\theta(Y_i \mid X_i, \omega_i)
= \phi \, V_\eta\{\mu_\theta(X_i)\}$. The monotonically increasing function
$g(\cdot)$ is referred to as the \emph{link function}
\citep{mcculaugh1989generalized}, and it is assumed to be known. Temporarily, we
will also assume that $\phi$ and $\eta$ are known.

Our point of departure is to use the quasi-likelihood to construct a
\emph{quasi-posterior} distribution for $r$:
\begin{align*}
  \Pit(dr \mid \phi, \eta, \Data)
  =
  \frac{Q(\theta) \, \Pi(dr)}{\int Q(\theta) \, \Pi(r) \ dr}
  \qquad \text{where} \qquad
  Q(\theta) = \exp\left\{ 
    -\frac{1}{2\phi} \sum_i \omega_i D\{Y_i, \mu_\theta(X_i)\} 
  \right\}.
\end{align*}
To remove restrictive assumptions on $\mu_\theta(x)$, we specify a nonparametric
prior on $r$; we focus here on Bayesian additive regression trees (BART), but in
principle there is no obstruction to modeling $r(x)$ with, for example, a
Gaussian process \citep{rasmussen2005gaussian}.

\begin{table}
  \centering
\begin{tabular}{lrrr}
  \toprule 
  Model                   & $\omega$ & $\mu$                                 & $V(\mu)$                \\
  \midrule
  Quasi-binomial          & $n_{i}$  & $(1+e^{r})^{-1}$                      & $\mu(1-\mu)$            \\
  Quasi-Poisson           & 1        & $e^{r}$                               & $\mu$                   \\
  Quasi-gamma             & 1        & $e^{-r} \text{ or } e^{r}$                              & $\mu^{2}$               \\
  Quasi-power             & 1        & $e^{r}$                               & $\mu^{k}$               \\
  Quasi-multinomial       & $n_{i}$  & $\frac{e^{r_{j}}}{\sum_{k}e^{r_{k}}}$ & $D_{\mu}-\mu\mu^{\top}$ \\
  \bottomrule
\end{tabular}
\caption{Quasi-likelihood and models for which we can provide simple MCMC
  algorithms when $r(x)$ is modeled with BART. Model gives the quasi-likelihood
  model name, $\omega$ denotes the form of the weight, $\mu$ denotes the formula
  for $\mu_\theta(X_i)$ in terms of $r(X_i)$, and $V(\mu)$ denotes the variance
  function. For the quasi-multinomial, $D_\mu = \diag(\mu_1, \ldots, \mu_K)$
  where $K$ is the dimension of $Y_i$. \label{tab:models}}
\end{table}

In Table~\ref{tab:models} we provide a list of models that we will consider,
with the different choice of weights $\omega_i$, choice of $g^{-1}$ to relate
$\mu_\theta(x)$ to $r(x)$, and variance function. Before proceeding, we briefly
describe the common use cases for each of these models. The quasi-binomial and
quasi-Poisson models are generally used to address overdispersion \citep[see,
for example][Chapter 4.5 and 5.5]{mcculaugh1989generalized} for binomial-type
and count-type data respectively, as the binomial and Poisson distributions
directly imply that $\phi = 1$; additionally, the beta distribution satisfies
the quasi-binomial variance relation, and so the quasi-binomial model can also
be used to perform beta regression. The quasi-gamma model is in fact formally
equivalent to the usual gamma regression model, but we continue with the term
``quasi-gamma'' to make clear first that we are not assuming that the underlying
data is gamma distributed and second that the dispersion parameter $\phi$ is not
estimated using the gamma likelihood; beyond the gamma distribution, the
quasi-gamma model naturally addresses scale regression models of the form $Y_i =
e^{r(X_i)} \, \epsilon_i$ with $\E(\epsilon_i) = 1$ and $\Var(\epsilon_i) =
\phi$, and so is valid for many settings of interest where the data are not
gamma distributed. The quasi-power model generalizes both the quasi-Poisson and
quasi-gamma models to allow for an arbitrary power in the variance, and is
appropriate for the medical expenditure data we analyze here; these are
sometimes referred to as \emph{Tweedie models} \citep{smyth1999adjusted}.
Finally, the quasi-multinomial model is appropriate both for multinomial-type
data to handle overdispersed outcomes, but is also appropriate for
simplex-valued outcome data such as Dirichlet-outcome data, where outcomes will
tend to be \emph{underdispersed} relative to the multinomial distribution.

\subsection{Review Bayesian Additive Regression Tree}

Conveniently, when the \emph{Bayesian additive regression tress} (BART)
framework is used to model $r(x)$, all of the quasi-likelihood models given in
Table~\ref{tab:models} are amenable to simple Gibbs sampling algorithms using
``Bayesian backfitting'' \citep{hastie2000bayesian,chipman2010bart}. The BART we
framework we use sets
\begin{align*}
  \mu_\theta(x) = g^{-1}\{r(x)\} \qquad \text{where} \qquad
  r(x) = \sum_{t = 1}^T \Treef(x; \Tree_t, \sM_t).
\end{align*}
The function $\Treef(x; \Tree, \sM)$ denotes a \emph{regression tree} with
decision tree $\Tree$ and leaf node parameters $\sM$, which returns the value
$\lambda_{\ell}$ if $x$ is associated to leaf $\ell$ of tree $\Tree$. A schematic of
such a regression tree is given in Figure~\ref{fig:tree}, with $x = (x_1,
x_2)^\top \in [0,1]^2$. It is apparent from this schematic that a regression
tree on $[0,1]^P$ is associated with a piecewise constant function, with $r(x)$
corresponding to a basis function expansion of many such decision shallow trees.

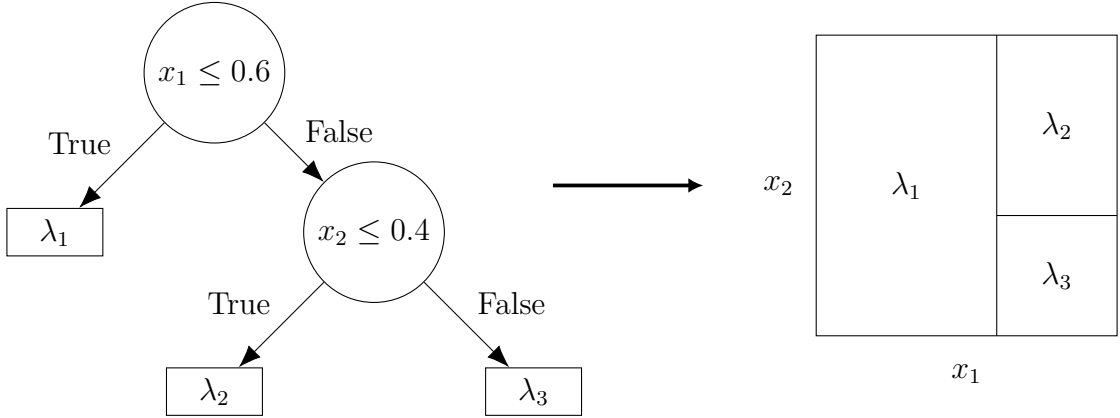
\begin{figure}
  \centering
  \begin{tikzpicture}[node distance=3cm, 
    decision/.style={circle, draw, text width=1.5cm, align=center},
    outcome/.style={rectangle, draw, text width=1cm, align=center}]

    \tikzset{
      big arrow/.style={
        draw,
        line width=0.5mm, 
        -{Latex[length=2mm,width=2mm]}, 
      }
    }

    \node[decision] (d1) {\(x_1 \leq 0.6\)};
    \node[decision, below right of=d1] (d2) {\(x_2 \leq 0.4\)};

    \node[outcome, below left of=d1] (o1) {\(\lambda_1\)};
    \node[outcome, below left of=d2] (o2) {\(\lambda_2\)};
    \node[outcome, below right of=d2] (o3) {\(\lambda_3\)};

    \draw[-{Latex[length=3mm]}] (d1) -- node[above left] {True} (o1);
    \draw[-{Latex[length=3mm]}] (d1) -- node[above right] {False} (d2);
    \draw[-{Latex[length=3mm]}] (d2) -- node[above left] {True} (o2);
    \draw[-{Latex[length=3mm]}] (d2) -- node[above right] {False} (o3);

    \draw[big arrow] (4.5,-1.5) -- (6.5,-1.5);

    \draw (8,0.5) -- (8,-3.5) -- (12,-3.5) -- (12,0.5) -- cycle;
    \draw (10.4, 0.5) -- (10.4, -3.5);
    \draw (10.4, -1.9) -- (12, -1.9);

    \node at (9.2, -1.5) {$\lambda_1$};
    \node at (11.2, -0.7) {$\lambda_2$};
    \node at (11.2, -2.7) {$\lambda_3$};
    \node at (10, -4) {$x_1$};
    \node at (8-0.5, -1.5) {$x_2$};

    
  \end{tikzpicture}
  \caption{Schematic of a regression tree decision tree (left) and the induced
    step-function on $[0,1]^2$ induced by the tree.}
  \label{fig:tree}
\end{figure}

BART is similar to other decision tree ensembling approaches, such as gradient
boosting \citep{friedman2001greedy}, with the main differences being (i) that we
place prior distributions on the $(\Tree_t, \sM_t)$'s and (ii) that inference
proceeds via Markov chain Monte Carlo rather than in a greedy fashion. We favor
the use of BART as a default prior for $r(\cdot)$, as it is highly flexible,
performs very well in practice, and is admits ``default priors'' that allow the
analyst to automatically choose the model hyperparameters; see, for example,
\citet{chipman2010bart} or \citet{linero2018bayesian} for the performance of
BART on benchmark datasets relative to random forests, gradient boosting, the
lasso. BART also possesses a number of desirable theoretical properties, such as
obtaining near-minimax rates of posterior concentration within the space of
functions that consist primarily of main effects and low-order interactions
\citep{rockova2020posterior,linero2018bayesian}. See \citet{hill2020bayesian}
and \citet{linero2017review} for comprehensive reviews on BART and Bayesian
decision tree-based methods.

\subsection{Prior Specification and Bayesian Backfitting Algorithms}
\label{sec:prior-specification}

Assuming that $V_\eta(\cdot)$ and $\phi$ are fixed a-priori, our approach only
requires specifying a prior for the $(\Tree_t, \sM_t)$'s. We use the default
branching process prior on the $\Tree_t$'s described by \citet{chipman2010bart},
which encourages most decision trees to be of depth one or two; we defer the
interested reader to that work for details. These priors for $\Tree_t$ are
empirically robust choices, and are taken by default by most BART software
packages \citep{kapelner2016bart,dorie2023dbarts,sparapani2021nonparametric}. In
our illustrations and software we also use the Dirichlet splitting rule prior
introduced by \citet{linero2018abayesian} to allow the model to automatically
remove irrelevant predictors.

The only remaining detail for prior specification is the choice of prior for the
$\sM_t$'s. We assume that the $\lambda_{t\ell}$'s are iid from some distribution
$\pi_\lambda(\lambda)$. In order to obtain a conditionally conjugate model, we
take $\pi_\lambda(\lambda)$ to be a $\log \Gam(a,b)$ distribution for the
quasi-binomial, quasi-Poisson, quasi-gamma, and quasi-multinomial models, with
$a$ and $b$ chosen so that $\E(\lambda) = 0$ and $\Var(\lambda) =
\sigma^2_\lambda$ for some user-specified value of $\sigma_\lambda$. We then
follow \citet{chipman2010bart} and set $\sigma_\lambda = 3 / (k \sqrt T)$ with
$k = 2$ by default; see \citet{chipman2010bart} for a justification of this
choice, and in particular for the normalization in the prior by $\sqrt T$.
Because it is unclear that the performance of the defaults extend to the new
models we propose, we also recommend using a half-Cauchy distribution with scale
$3 / (k \sqrt T)$ as a prior for $\sigma_\lambda$; this allows the model to
estimate $\sigma_\lambda$, while encouraging it to be near the default of
\citet{chipman2010bart}.

The models in Table~\ref{tab:models} can be fit via Markov chain Monte Carlo
using the generalized Bayesian backfitting algorithm given by
\citet{hill2020bayesian}, with the exception of the quasi-power model. These
algorithms work by first approximately sampling $\Tree_t$ from its conditional
distribution $[\Tree_t \mid \Tree_{-t}, \sM_{-t}, \phi, \eta, \Data]$ using a
Metropolis-Hastings algorithm and second sampling the leaf node parameters
$\sM_t$ from the full conditional distribution $[\sM_t \mid \Tree_t, \Tree_{-t},
\sM_{-t}, \phi, \eta, \Data]$; here, $\Tree_{-t}$ and $\sM_{-t}$ denote the
decision trees and leaf node parameters of all of the regression trees except
tree $t$. Gibbs samplers for these models are given in the appendix.

By contrast, the quasi-power model cannot be fit using any simple extensions of
existing BART sampling algorithms. We instead use the Laplace approximation,
which can be computed analytically, to approximate the relevant marginal
likelihoods and full conditionals. After applying the Laplace approximation,
implementing the quasi-power model is no more computationally intensive than the
original BART model. Moreover, as shown by \citet{linero2022generalized}, it is
easy to use the Laplace approximation to construct a valid proposal in a
Metropolis-Hastings algorithm that holds the target posterior exactly invariant.
Details of this approximation are given in the appendix.

\subsection{Extension to Multivariate Quasi-Likelihood Models}

Defining the quasi-multinomial model in Table~\ref{tab:models} requires
extending the quasi-likelihood to the multivariate setting \citep[][Section
9.3]{mcculaugh1989generalized}. Consider the pair of moment equations
$\E_\theta(Y_i \mid X_i = x, \omega_i) = \mu_\theta(x)$ and $\Var_\theta(Y_i
\mid X_i = x, \omega_i) = \frac{\phi}{\omega_i} \, V_\eta\{\mu_\theta(x)\}$
where $Y_i$ and $\mu_\theta(x)$ are now a $K$-dimensional vectors and
$V_\eta(\mu)$ is a $K \times K$ covariance matrix. The \emph{multivariate
  quasi-likelihood} (MQL) is defined as
\begin{align*}
  q_\theta(y \mid x, \omega)
  = \exp\left\{ -\frac{\omega \, D\{y, \mu_\theta(x)\}}{2\phi} \right\}
  \qquad \text{where} \qquad
  D(y, \mu) = -2 \int_{\mu}^y (y - t)^\top V_\eta(t)^{-1} \ dt,
\end{align*}
where both $\mu$ and $y$ are $K$-dimensional vectors, $V(t)$ is a $K \times K$
covariate matrix, and the integral is a path-independent line integral.

The MQL model we consider takes $g\{\mu_\theta(X_i)\} = r(X_i)$ where $r(x) =
\sum_{t = 1}^T \Treef(x; \Tree_t, \sM_t)$ is now a sum of \emph{multivariate
  decision trees}; see, for example, \citet{linero2020semiparametric} for how
multivariate decision trees can be combined using BART to model multivariate
outcomes.

We focus on the quasi-multinomial model, which takes $V(\mu) = D_\mu - \mu
\mu^\top$ where $D_\mu = \diag(\mu_1, \ldots, \mu_K)$ is a diagonal matrix. The
quasi-likelihood for this model is
\begin{align*}
  q_\theta(y \mid x, n)
  = \frac{\exp\left\{ z^\top r(x) / \phi \right\}}
  {[\sum_k \exp\{r_k(x)\}]^{n / \phi}}
\end{align*}
where $\omega = n$ corresponds to the sample size of a multinomial count
vector $z$ and $y = z / n$ is a vector of proportions. This model reduces to the
usual multinomial likelihood in the special case when $\phi = 1$, and can be
used to also model Dirichlet response data when $n = 1$ and $\phi < 1$. Note
that, when applied to multinomial data, the differences $r_k(x) - r_j(x)$
correspond to the log-odds ratio associated to a comparison between category
level $k$ and category level $j$.

For a general MQL model, assuming the $\mu_i$'s are known a moment estimator of
$\phi$ is given by
\begin{align}
  \widehat \phi = \frac{1}{NK} \sum_i
  \omega_i (Y_i - \mu_i)^\top V(\mu_i)^{-1} (Y_i - \mu_i).
\end{align}
In the special case of the quasi-multinomial model, this moment estimator
becomes
\begin{align*}
  \widehat\phi = \frac{1}{N(K-1)} \sum_i \sum_k
    \frac{n_i (Y_{ik} - \mu_{ik})^2}{\mu_{ik}}
\end{align*}
with the modification of $K$ to $K - 1$ due to the loss of a degree of freedom
from the constraint $\sum_k Y_{ik} = 1$; it is straight-forward to verify that
$\widehat\phi$ is unbiased for $\phi$.

To obtain a Bayesian nonparametric quasi-multinomial model, we set $r = r(x)$ in
the above model where $r(x)$ is a sum of multivariate decision trees $r(x) =
\sum_t \Treef(x; \Tree_t, \sM_t)$ where $\Treef(x ; \Tree , \sM) = (\lambda_{\ell 1},
\ldots, \lambda_{\ell K})^\top$ if $x$ is associated to leaf node $\ell$ of
$\Tree$. For identifiability purposes in the quasi-multinomial model, it is
customary to set $r_1(x) \equiv 0$. However, as the individual functions
$r_k(x)$ are not of direct interest, and to ensure that the model is symmetric
across all categories, we place no restrictions on the $\lambda_{\ell k}$'s. The
$\lambda_{\ell k}$'s are given the same log-gamma prior described in
Section~\ref{sec:prior-specification}, with the default $\sigma_\lambda \sim
\text{half-Cauchy}\{0, 3 / (k \sqrt{2T})\}$ with $k = 2$ by default. The choice
$3 / k \sqrt {2T}$ is chosen to represent a prior belief that, for each $x$ the
log-odds ratio $r_k(x) - r_j(x)$ is approximately normally distributed with mean
$0$ and standard deviation $3/k$.

\section{Inference for the Dispersion Parameter}
\label{sec:inference-for-the-dispersion-parameter}


We now describe some strategies for estimating $\phi$ and discuss their merits.
For the sake of exposition, the updates are roughly ordered from least to most
recommended.

\begin{paragraph}{Extended Quasi-Posteriors (EQP)}
  Introduced by \citet{weddernburn1974quasi}, the \emph{extended
    quasi-likelihood} (EQL) augments $Q(\theta)$ with an additional term
  $\sqrt{2\pi \phi V_\eta(Y_i)}$ to obtain
  \begin{align*}
    Q^+(\theta) = \prod_i \frac{1}{\sqrt{2\pi \phi V_\eta(Y_i)}}
      \exp\left( -\frac{\omega_i D\{Y_i, \mu_\theta(X_i)\}}{2\phi} \right).
  \end{align*}
  The extended quasi-posterior is then given by $\Pi(dr, d\phi, d\eta \mid
  \Data) \propto Q^+(\theta) \, \Pi(dr, d\phi, d\eta)$. Given a gamma prior
  $\phi^{-1} \sim \Gam(a, b)$, the EQP is convenient in that the full
  conditional for $\phi$ is given by $\phi^{-1} \sim \Gam\{a + N / 2, b + \sum_i
  \omega_i D\{Y_i, \mu_\theta(X_i)\} / 2\}$. The EQP is not recommended because it
  generally does not produce a consistent estimator of $\phi$
  \citep{davidian1988note}. The primary advantage of EQL is that it has none of
  the compatibility issues of the PLP and BBQ approaches discussed below.
\end{paragraph}

\begin{paragraph}{Pseudo-Empirical Bayes}
  A simple strategy is to use the moment estimator \eqref{eqn:moment} obtained
  by fitting the model with (say) $\phi = 1$, and then refit the model with this
  new value of $\phi$. This might then be iterated: we refit the model with the
  new value of $\phi$, apply \eqref{eqn:moment}, and repeat this some modest
  number of times. We do not recommend this approach for Bayesian nonparametric
  models because it is not computationally efficient and we lack of obvious
  replacement for $P$ in \eqref{eqn:moment}. Instead, inspired by the
  pseudo-likelihood (PL) approach of \citet{davidian1988note}, we propose a
  pseudo-empirical Bayes strategy, which iteratively updates $(\phi, \eta)$ by
  setting
  \begin{align}
    \begin{split}
    \label{eq:EB}
    (\phi_{t+1}, \eta_{t+1}) &\gets \arg \max_{\phi, \eta} \int \text{PL}(\phi, \eta) \ 
    \ \Pit(r \mid \phi_t, \eta_t, \Data)
    \qquad \text{where} \qquad \\
    \text{PL}(\phi, \eta) &= 
    \prod_i \frac{1}{\sqrt{2\pi\phi\,V_\eta\{\mu(X_i)\}/\omega_i}}
    \exp\left[-\frac{\omega_i\{Y_i - \mu(X_i)\}^2}
    {2 \, \phi \, V_\eta\{\mu(X_i)\}} \right].
    \end{split}
  \end{align}
  To approximately evaluate the integral, one can use empirical Bayes Gibbs
  sampling \citep{casella2001empirical}. This can be iterated a modest number of
  times to obtain plausible values of $(\phi, \eta)$. This is computationally
  expensive and does not provide uncertainty quantification for $(\phi, \eta)$,
  and so still not ideal.
\end{paragraph}

\begin{paragraph}{Pseudo-Likelihood Posteriors (PLP)}
  The pseudo-likelihood posterior (PLP) approach extends the pseudo-empirical
  Bayes approach to obtain a posterior distribution for $(\phi, \eta)$, and is
  both computationally simpler and obtains (uncalibrated) uncertainty
  quantification. The main idea is to run Bayesian backfitting on $r(x)$, but
  include updates for $(\phi, \eta)$ based on the partial likelihood; that is,
  we use a two-stage Gibbs sampler:
  \begin{enumerate}
  \item Sample $r$ from $\Pit(dr \mid \phi, \eta, \Data) \propto
    Q(\theta) \, \Pi(dr)$.
  \item Sample $(\phi, \eta)$ from $\Pit(d\phi, d\eta \mid r, \Data) \propto
    \text{PL}(\phi, \eta) \, \Pi(d\phi, d\eta)$.
  \end{enumerate}
  As with the EQL approach, under the gamma prior $\phi^{-1} \sim \Gam(a,b)$,
  the PL approach results in a simple full conditional $\phi^{-1} \sim \Gam(a +
  N/2, b + \sum_i Z_i^2 / 2)$, where $Z_i = \frac{\omega_i^{1/2}\{Y_i -
    \mu_\theta(X_i)\}}{V\{\mu_\theta(X_i)\}^{1/2}}$ is the \emph{standardized
    residual}. The PL approach generally produces good point estimates of
  $(\phi, \eta)$ and uncertainty quantification for $r(\cdot)$, but may not
  produce good uncertainty quantification for $(\phi, \eta)$ and so should be
  used with caution when uncertainty in the variance function is of intrinsic
  interest. The PL approach also samples from incompatible conditionals for $r$
  and $(\phi, \eta)$ in the sense that there is no joint distribution for which
  the update distributions are the conditionals.
\end{paragraph}

\begin{paragraph}{Bayesian Bootstrap Quasi-Likelihood (BBQ)}
  A simple, robust, approach to sampling $\phi$ is to use an approach inspired
  by the loss-likelihood bootstrap \citep{lyddon2019general}. Again let $Z_i =
  \frac{\omega_i^{1/2}\{Y_i - \mu_\theta(X_i)\}}{V\{\mu_\theta(X_i)\}^{1/2}}$
  denote the standardized residual and note that $\E_\theta(Z_i^2) = \phi$. We
  model the marginal distribution of the $Z_i$'s using a Bayesian bootstrap
  $F_Z(dz) = \sum_{i = 1}^N p_i \, \delta_{Z_i}(dz)$ with the full conditional
  $p \sim \Dirichlet(1, \ldots, 1)$ and set
  \begin{align}
    \label{eqn:bbq}
    \phi = \sum_{i = 1}^N p_i \, Z_i^2.
  \end{align}
  MCMC proceeds as with the PL approach, except that $\phi$ is updated by
  sampling $p \sim \Dirichlet(1,\ldots,1)$ and using \eqref{eqn:bbq}. This
  approach can further accommodate estimation of $\eta$ by replacing
  \eqref{eqn:bbq} with the update $(\phi, \eta) = \arg \max_{\phi, \eta} \sum_i
  p_i \log \Normal\big(Y_i \mid \mu_\theta(X_i), \phi \,
  V_\eta\{\mu_\theta(X_i)\} / \omega_i\big)$. 
  The advantage of the BBQ approach over the PLP approach is that we intuitively
  expect the use of the bootstrap to provide better calibrated uncertainty
  quantifification for $(\phi, \eta)$. The downsides of the BBQ approach is that
  the conditionals of $r$ and $(\phi, \eta)$ are again incompatible and that it
  may be difficult to jointly optimize $(\phi, \eta)$ in a robust fashion.
\end{paragraph}

\section{Theoretical Properties}
\label{sec:theoretical-properties}

We now examine the theoretical properties of the PL approach
\citep[see][for results for the quasi-posterior with known
\(\phi\)]{agnoletto2023bayesian}. Proofs of these results can be found in the
appendix. For simplicity, we consider the parametric setting with $r(x) =
x^\top\beta$ and assume $\eta$ to be known. Let $g(\phi \mid \beta)$ denote
conditional density that $\phi$ is sampled from in our Gibbs sampler and let
$\pit(\beta \mid \phi) \propto Q(\beta, \phi) \, \pi(\beta)$. We restrict
attention to the ideal \emph{two-step Gibbs sampler}:
\begin{enumerate}
\item[i.] Sample $\beta \sim \pit(\beta \mid \phi)$.
\item[ii.] Sample $\phi \sim g(\phi \mid \beta)$ where $g$ is the PL update
  density for $\phi$.
\end{enumerate}
We show that the sequence $\beta^{(1)}, \beta^{(2)}, \cdots$ is a Markov chain
with a stationary distribution that is expressible as a mixture $\pit(\beta) =
\int \pit(\beta \mid \phi) \, m(\phi) \ d\phi$. The mixing distribution
$m(\phi)$ is also the marginal distribution of $\phi$ in the stationary
distribution, and is well-behaved in the sense that it concentrates near
$\phi_0$, which effectively causes $\pit(\beta)$ to be well-approximated by
a $\pit(\beta \mid \phi_0)$ density, leading to a Bernstein-von Mises result for
$\pit(\beta)$ due to existing results of \citet{agnoletto2023bayesian} and
\citet{miller2021asymptotic}.

A minor inconvenience with the BBQ and PL approaches is that there is typically
no joint distribution $\pit(\beta, \phi)$ whose conditional distributions are
$\pit(\beta \mid \phi)$ and $g(\phi \mid \beta)$; a necessary condition for this
to occur is that factorization $\pit(\beta \mid \phi) / g(\phi \mid \beta) =
u(\beta) \, v(\phi)$ holds for some functions $u(\cdot), v(\cdot)$
\citep[][Theorem 4.1]{arnold1989compatible}, which is not the case for any of
the models we have discussed. The following result states that, despite these
challenges, the chain has a well-defined stationary distribution.

\begin{theorem}
  \label{thm:stationary}
  Suppose that Condition A and Condition P in
  Appendix~\ref{sec:regularity-conditions} hold. Then the two-step Gibbs sampler
  initialized at any point $\phi > 0$ has a unique stationary distribution with
  density $\pit(\beta, \phi) = \pit(\beta) \, g(\phi \mid \beta)$ where
  $\pit(\beta) = \int \pit(\beta \mid \phi) \, m(\phi) \ d\phi$ for some mixing
  distribution $m(\phi)$. Additionally, $m(\phi) = \int \pit(\beta) \, g(\phi
  \mid \beta) \, d\beta$.
\end{theorem}

With the issue of the existence of the stationary distribution dispensed with,
we now present our Bernstein-von Mises result.

\begin{theorem}
  \label{thm:bvm}
  Under Condition P, Condition A, Condition F, Conditions D1 -- D4, and
  Condition G listed in the Appendix hold, we have
  \begin{align*}
    \int |\pit(\beta) - N\{\beta \mid \widehat \beta, \phi_0 N^{-1} H(\beta_0)^{-1}\}|
    \ d\beta \to 0 \qquad \text{almost surely},
  \end{align*}
  where $H(\beta_0)$ is given in Condition D and $\phi_0 N^{-1}
  H(\beta_0)^{-1}$is the asymptotic covariance matrix of the maximum
  quasi-likelihood estimator $\widehat \beta = \arg \min_\beta \sum_i \omega_i
  D\{Y_i ; g^{-1}(X_i^\top \beta)\}$. Additionally, $m(\phi)$ converges in
  distribution to a point-mass at $\phi_0$ almost surely.
\end{theorem}

When the model is correctly specified, Theorem~\ref{thm:bvm} implies correct
asymptotic coverage of posterior credible intervals for $\beta$. While we do not
prove this, under misspecification we generally expect that a version of
Theorem~\ref{thm:bvm} will hold, but $\phi_0 \, H(\beta_0)^{-1} / N$ will no
longer be the asymptotic covariance matrix of $\widehat\beta$, and so inferences
will not be correctly calibrated.

\begin{remark}
  Theorem~\ref{thm:bvm} is proved for the PLP rather than the BBQ posterior
  simply because it is more analytically tractable, however we expect
  Theorem~\ref{thm:bvm} to also hold for the BBQ posterior. As argued by
  \citet{davidian1988note}, the extended quasi-likelihood does not lead to
  consistent estimates of $(\phi, \eta)$; in our own experiments, it was instead
  observed that the EQP concentrates around $\beta_0$ but does not possess
  good Frequentist properties, and that the posteriors for $\phi$ and $\eta$ do
  not concentrate around their true values. This is unfortunate because the EQP
  avoids the problems listed above: inference for all components of $\theta$ is
  based on the single likelihood $Q^+(\theta)$, and so the stationary
  distribution of the Markov chain will exist and be equal to $Q^+(\theta) \,
  \Pi(d\theta) / \int Q^+(\theta) \, \Pi(d\theta)$ so long as the denominator is
  finite. Because of these issues, we do not study the theoretical properties of
  the EQP posterior.
\end{remark}

\begin{remark}
  The regularity conditions used in the appendix are optimized for ease of proof
  rather than generality, and our results should be regarded primarily as a
  proof-of-concept. The most substantial simplification we make is that $\phi_0$
  is assumed to lie in an interval $[a,b]$, and $\beta_0$ in some bounded set,
  with both sets known a-priori.
\end{remark}

\section{Illustrations}
\label{sec:illustrations}

\subsection{Experiments With Parametric Models}
\label{sec:bernstein-von-mises}

We first provide some numerical support for the Bernstein-von Mises results
given in Section~\ref{sec:theoretical-properties}. That a Bernstein-von Mises
result will hold does not follow from our theory, even for parametric models,
because we usually cannot use the ideal two-step Gibbs sampler. We consider data
generated according to the quasi-Poisson model with
\begin{align*}
  Y_i = \phi \, Z_i \qquad \text{where} \qquad 
  Z_i \sim \Poisson\{\mu_\theta(X_i) / \phi\}
\end{align*}
and $\mu_\theta(x) = \exp(\beta_0 + \beta_1 \, x)$. Data generated according to
this model satisfies the quasi-Poisson relation $\Var(Y_i \mid X_i = x) = \phi
\, \mu_\theta(x)$. We considered a flat prior $\pi(\beta) \propto 1$, $X_i \sim
\Normal(0,1)$, $N = 2000$, and use the BBQ approach to quantify uncertainty in
$\phi$. Hamiltonian Monte Carlo \citep[HMC,][]{neal2011mcmc} was used to sample
from the quasi-posterior.

\begin{figure}
  \centering
  \includegraphics[width = .7\textwidth]{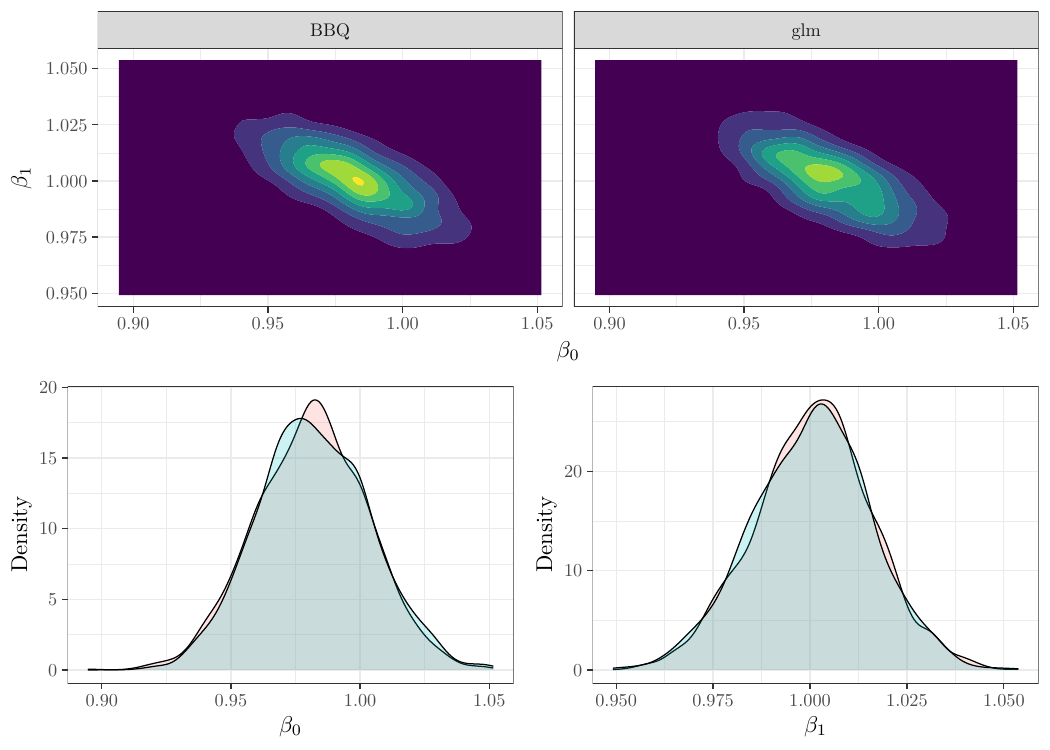}
  \caption{Joint and marginal posterior densities of the parameters from the
    quasi-Poisson simulation in Section~\ref{sec:bernstein-von-mises}. The
    quasi-posterior under the BBQ prior is given by ``BBQ'' while the asymptotic
    distribution with mean equal to the MLE and variance equal to the inverse
    Fisher information corresponds to ``glm.''}
  \label{fig:bvm}
\end{figure}

The posterior distribution is displayed in Figure~\ref{fig:bvm} and compared
with the asymptotic normal distribution obtained from standard quasi-likelihood
theory using the pilot estimator $\widehat\phi = (N-2)^{-1} \sum_i (Y_i -
\widehat \mu_i)^2 / \widehat\mu_i$. We see strong agreement between the
quasi-posterior and the asymptotic normal distribution, giving some
further evidence that the good theoretical properties of the two-step Gibbs
sampler likely extend to methods that replace sampling $\pit(\beta\mid\phi)$
with sampling a Markov transition density that merely leaves
$\pit(\beta\mid\phi)$ invariant.

Next, we consider a quasi-power model with $V_\eta(\mu) = \mu^\kappa$. This
model contains the quasi-Poisson $(\eta = 1)$ and quasi-gamma $(\eta = 2)$ as
submodels. We set $\mu_\theta(x) = \exp(x^\top\beta)$ with $\beta_j = 1.2 /
\sqrt 5$ for the $\kappa = 1$ setting and $\beta_j = 2 / \sqrt{5}$ for the
$\kappa = 2$ setting; these values were chosen to highlight differences between
the quasi-gamma and quasi-Poisson models without being unrealistic. We evaluate
the following models:
\begin{itemize}
\item \textbf{Quasi-Poisson:} a quasi-Poisson GLM fit using the \texttt{glm}
  function in \texttt{R} via maximum quasi-likelihood.
\item \textbf{Quasi-gamma:} a quasi-gamma GLM fit using the \texttt{glm}
  function in \texttt{R} via maximum quasi-likelihood.
\item \textbf{Bayesian Bootstrap Poisson:} samples of $\beta$ are collected by
  first sampling $p \sim \Dirichlet(1, \ldots, 1)$ and then
  solving the equation $\sum_i \frac{p_i (Y_i - \mu_\beta(X_i))}{\mu_\beta(X_i)}
    \, X_i = 0$.
\end{itemize}

\begin{figure}
  \centering
  \includegraphics[width=1\textwidth]{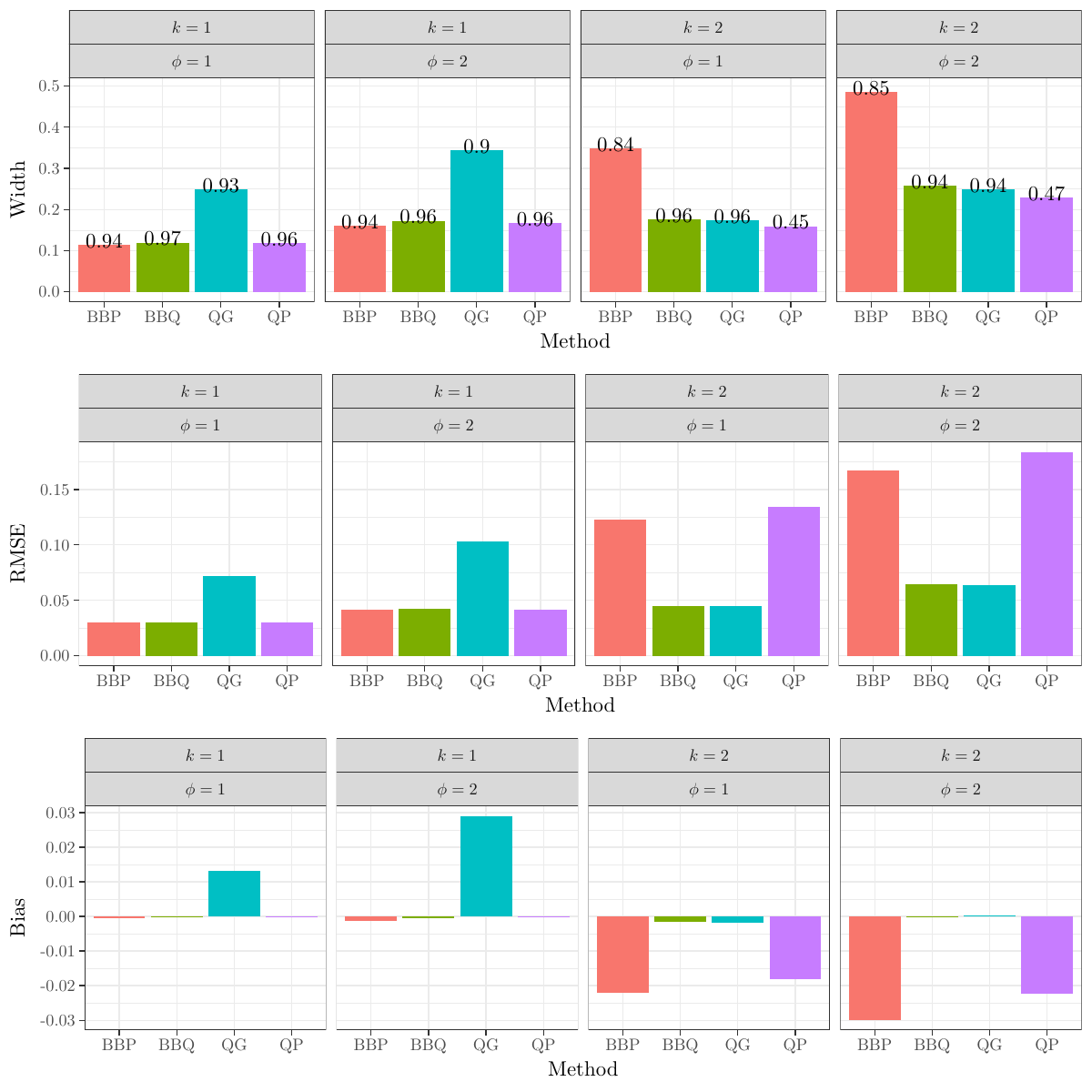}
  \caption{Width of 95\% intervals (top), root mean-squared error (middle), and
    bias (bottom) results for the power-variance experiment. BBP denotes the
    Bayesian Bootstrap Poisson model, BBQ denotes Bayesian Bootstrap
    Quasi-likelihood, QG denotes quasi-gamma, and QP denotes quasi-Poisson.
    Coverage rates are given above the widths in the top
    row. \label{fig:bbq_experiment}}
\end{figure}

Figure~\ref{fig:bbq_experiment} displays the root mean squared error, estimated
bias, and interval width/coverage associated to nominal 95\% interval estimates.
We see that the BBQ approach is the only approach that performs uniformly well
across all settings, and that the performance of all of the methods is not
greatly impacted by $\phi$. When $\kappa = 1$ (quasi-Poisson holds) we see that
the BBQ, Quasi-Poisson, and BB Poisson methods are all roughly equivalent, while
the quasi-gamma estimates have higher mean-squared errors and have substantially
wider intervals to obtain coverage rates just below the nominal. On the other
hand, when $\kappa = 2$ (quasi-gamma holds), the methods based on the Poisson
likelihood perform poorly in terms of interval width, coverage, RMSE, and bias,
while the BBQ approach essentially matches the performance of the quasi-gamma
model.

\subsection{Bayesian Nonparametric Quasi-Gamma Model}
\label{sec:bayesian-nonparametric-quasi-gamma}

We next consider a quasi-gamma model where $V(\mu) = \mu^2$ under a data
generating mechanism with $\mu_\theta(x) = e^{r(x)}$ where $r(x) = \sin(\pi \,
x_1 \, x_2) + 2 (x_3 - 0.5)^2 + x_4 + x_5 / 2$; the form of $r(x)$ is taken from
\citet{friedman1991multivariate}, where it was used to illustrate the benefits
of multivariate adaptive regression splines (MARS). For the outcome, we set
\begin{math}
  [Y_i^{-1} \mid X_i = x] \sim \Gam\{\alpha, (\alpha - 1) \, \mu_\theta(x)\}
\end{math}
which corresponds to a quasi-gamma model with $\phi = (\alpha - 2)^{-1}$. 

We compare the quasi-gamma model to the conventional (misspecified) gamma model
that takes $[Y_i \mid X_x = x] \sim \Gam\{\alpha, \alpha / \mu_\theta(x)\}$. We
note that the gamma and quasi-gamma models differ \emph{only} in how the
parameter $\phi$ is estimated, with the quasi-gamma model using BBQ and the
gamma model using the posterior derived from the gamma likelihood.

\begin{figure}
  \centering
  \includegraphics[width=.8\textwidth]{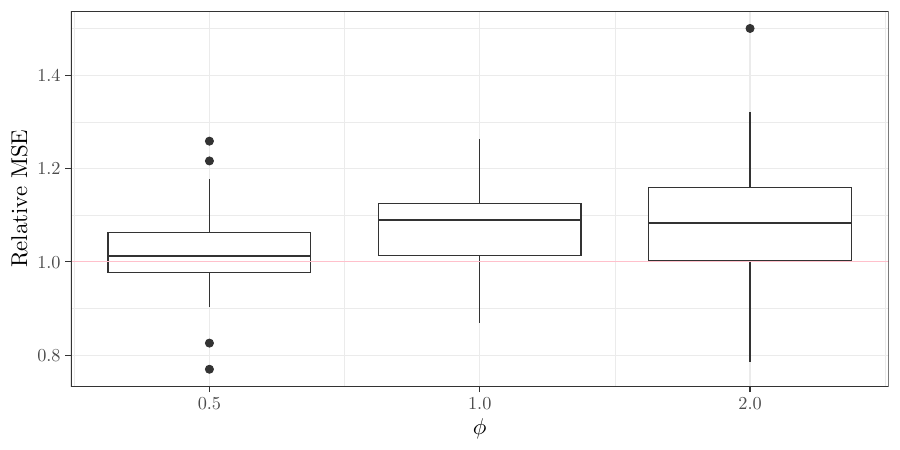}
  \caption{Boxplots of mean squared error of the gamma model relative to the quasi-gamma
    model for the experiment described in
    Section~\ref{sec:bayesian-nonparametric-quasi-gamma}.}
  \label{fig:invgam}
\end{figure}

The models were compared in terms of relative mean squared error, defined as
$\frac{\sum_i \{r(X_i) - \widehat r_{\text{G}}(X_i)\}^2}{\sum_i \{r(X_i) -
  \widehat r_{\text{QG}}\}^2}$, where $\widehat r_{\text{G}}(x)$ and $\widehat
r_{\text{QG}}(x)$ are the Bayes estimates of $r(x)$ under the gamma and
quasi-gamma models, respectively. The results of 200 replications of this
experiment under the settings $N = 250$, $P = 10$, and $\phi \in \{0.5, 1, 2\}$
are given in Figure~\ref{fig:invgam}. We see that the quasi-gamma model
typically outperforms the gamma model on this problem, with a larger gap
existing for larger values of $\phi$. In terms of practical differences, we
found that the gamma model tends to substantially underestimate $\phi$, while
BBQ tends to slightly overestimate $\phi$, with this accounting for the mild
differences in predictive performance.

\subsection{Bayesian Nonparametric Quasi-Power Model}
\label{sec:bayesian-nonparametric-quasi-power}

To assess the reliability of our implementation of the quasi-power model, we
use the quasi-power ground truth
\begin{math}
  [Y_i \mid X_i = x] \sim \Gam\{\alpha(x), \beta(x)\}
\end{math}
where $\alpha(x) = e^{r(x) / 2} / \phi$ and $\beta(x) = e^{-r(x) / 2} / \phi$;
this corresponds to a quasi-power model with $\mu_\theta(x) = e^{r(x)}$ and
$\kappa = 1.5$. We take $r(x)$ to be the same as in Section~\ref{sec:bayesian-nonparametric-quasi-gamma}.

We compare the quasi-power model to the quasi-Poisson and quasi-gamma models,
which correspond to special cases of the quasi-power model with $\kappa = 1$ and
$2$ respectively. We set $N = 250$, $P = 10$, $\phi \in \{0.5, 1, 2\}$, and
replicated the experiment $200$ times. Models were compared in terms of mean
squared error $N^{-1} \sum_i \{\mu_\theta(X_i) - \widehat \mu(X_i)\}^2$ where
$\widehat \mu(x)$ is the Bayes estimator of $\mu_\theta(x)$, as well as the
coverage per-experiment of the pointwise nominal 95\% credible intervals for the
$\mu_\theta(X_i)$'s and the average length of these intervals.

\begin{figure}
  \centering
  \includegraphics[width=.9\textwidth]{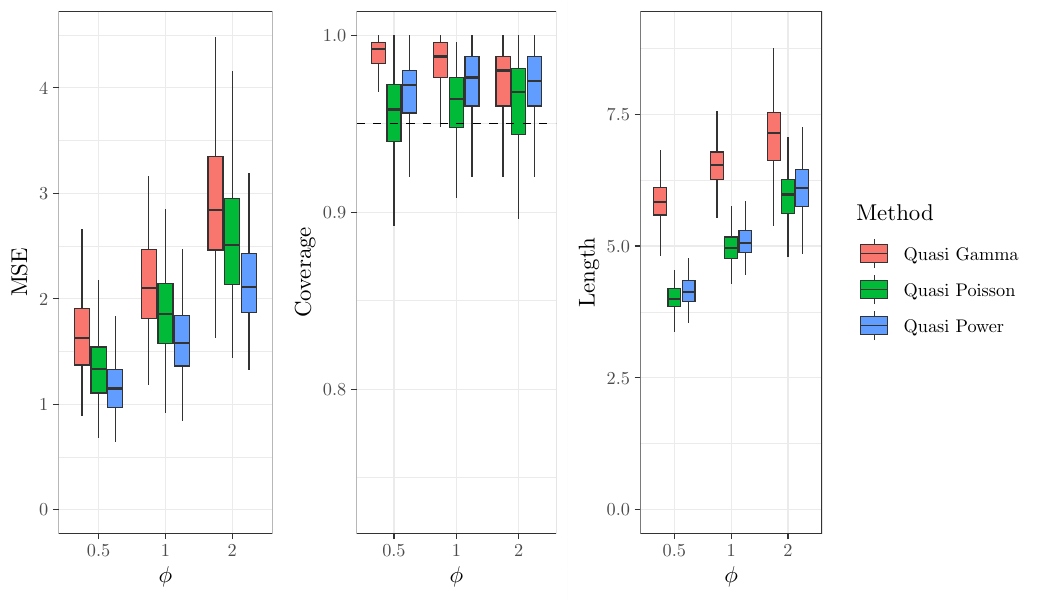}
  \caption{Results from simulation study of
    Section~\ref{sec:bayesian-nonparametric-quasi-power}. Boxplots displays the
    squared error, coverage, and interval length associated to $\mu(X_i)$ for $i
    = 1,\ldots, N$, averaged over $N$, for the 200 replications.}
  \label{fig:quasi-power_01}
\end{figure}

Results are given in Figure~\ref{fig:quasi-power_01}. As anticipated, the
quasi-power modestly improves upon the quasi-gamma and quasi-Poisson models in
terms of predictive performance. Interestingly, all models appear to have
conservative credible intervals, with the quasi-Poisson model coming the closest
to the nominal coverage levels.

\begin{figure}
  \centering
  \includegraphics[width=.9\textwidth]{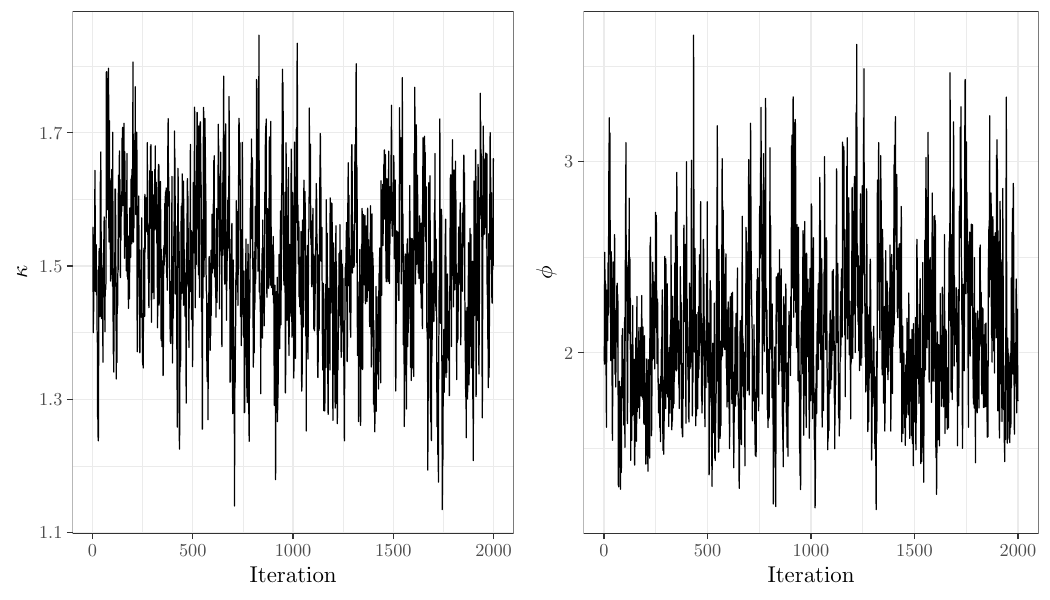}
  \caption{Mixing of $\phi$ and $\kappa$ for a single replication of the
    simulation experiment in
    Section~\ref{sec:bayesian-nonparametric-quasi-power} with $N = 2500$.}
  \label{fig:quasi-power_02}
\end{figure}

Figure~\ref{fig:quasi-power_02} displays the mixing of $\phi$ and $\kappa$ on a
single representative simulation iteration with $N = 2500$. We see that $\phi$
and $\kappa$ have posteriors centered around their true values, and mix
reasonable well over 2,000 post-burn-in samples.

\subsection{Quasi-Multinomial With Dirichlet Outcome Data}
\label{sec:quasi-multinomial-with-dirichlet-outcome-data}

Our last synthetic data experiment applies the quasi-multinomial model to
Dirichlet outcome data. We consider the data generating process $Y_i \sim
\Dirichlet\{\rho \, \mu_1(X_i), \rho \, \mu_2(X_i), \rho \, \mu_3(X_i)\}$ where
$\mu_k(x) \propto \exp\{r_k(x)\}$ and
\begin{align*}
  r_1(x) &= 2 x_1 + x_2, & r_2(x) &= x_1 + 4 x_2 x_3 & r_3 &= x_2 + 2 x_3.
\end{align*}
We consider $N = 1000$ observations, $X_i \sim \Uniform([0,1]^5)$, and set $\rho
= 0.5$, which corresponds to a quasi-multinomial model with $\phi = 2/3$.

\begin{figure}
  \centering
  \includegraphics[width=1\textwidth]{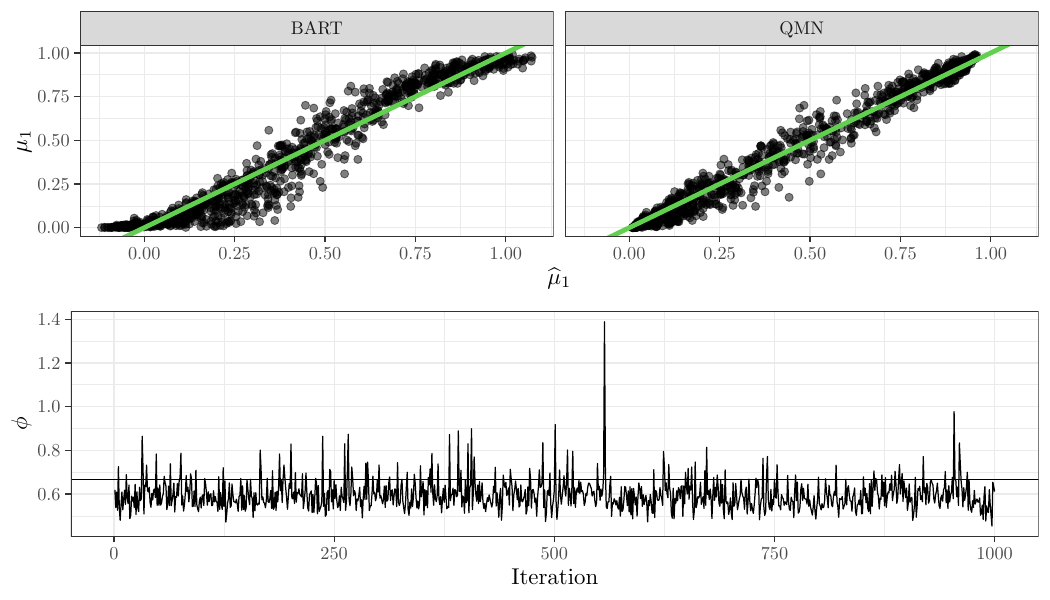}
  \caption{Top: comparison of the Bayes estimate ($x$-axis) of $\mu_1(X_i)$ to
    the true value of $\mu_1(X_i)$ ($y$-axis) for the BART normal regression
    model (left) and quasi-Multinomial model (right); the line $y = x$ is
    provided for comparison. Bottom: traceplot of
    $\phi$ for the quasi-Multinomial model.}
  \label{fig:qmn}
\end{figure}

Figure~\ref{fig:qmn} gives results for a single simulated dataset, comparing the
quasi-multinomial model to the standard BART regression model implemented in the
\texttt{dbarts} package that directly estimates $\mu_1(x)$; results for this
single simulated dataset are typical of replications. We see that the
quasi-multinomial model generally performs better than BART, with no obvious
bias and estimates falling quite close to the lines $x = y$.

\begin{table}
  \centering
  \input{Figure/qmn_results.tex}
  \caption{Results of the simulation experiment in Section~\ref{sec:quasi-multinomial-with-dirichlet-outcome-data}}
  \label{fig:qmn-results}
\end{table}

This experiment was replicated 100 times, and results are presented in
Table~\ref{fig:qmn-results}. Here, RMSE refers to the RMSE in estimating
$\mu_1(X_i)$ given by $\sqrt{\frac{1}{N} \sum_i \{\mu_1(X_i) - \widehat
  \mu_1(X_i)\}^2}$, Width refers to the average interval width of the 95\%
highest posterior density (HPD) credible interval, and Coverage refers to the
average coverage of the HPD interval. By properly accounting for the presence of
heteroskedasticity we see that the quasi-multinomial model outperforms BART on
all three metrics, with a much lower RMSE being attained and nominal coverage
being achieved with shorter intervals.

\subsection{The Medical Expenditure Panel Survey}
\label{sec:meps}

The medical expenditure panel survey (MEPS) is an ongoing survey on the usage of
the healthcare system in the United States by individuals, insurance companies,
and hospitals. We consider as an outcome $Y_i$ the total medical expenditure
incurred during 2011 by a subset of 9,426 adult women who incurred at least some
medical expenditure. Parametric quasi-likelihood methods have been used in the
past for this data and it has been demonstrated that the data is highly
heteroskedastic with a mean-variance relation well-described by a quasi-power
relation $V(\mu) = \mu^\kappa$ \citep[][who found $\kappa \approx 1.5$ on the
subset of the data they examined]{natarajan2008variance}. Our interest is in
determining what factors lead to higher forecasts of medical expenditures, with
variables of interest including demographic variables, measurements of an
individual's health, and direct measures of the usage of the healthcare system.
Predictors include age, insurance status, race, whether an individual is a
regular smoker, income, self-assessed physical health status, mental health
status, body mass index, number of prescription refills during the year,
education level, number of dentist trips during the year, and indicators of
major health diagnoses (diabetes, heart attack, stroke, cancer, and arthritis).

\begin{figure}
  \centering
  \includegraphics[width=1\textwidth]{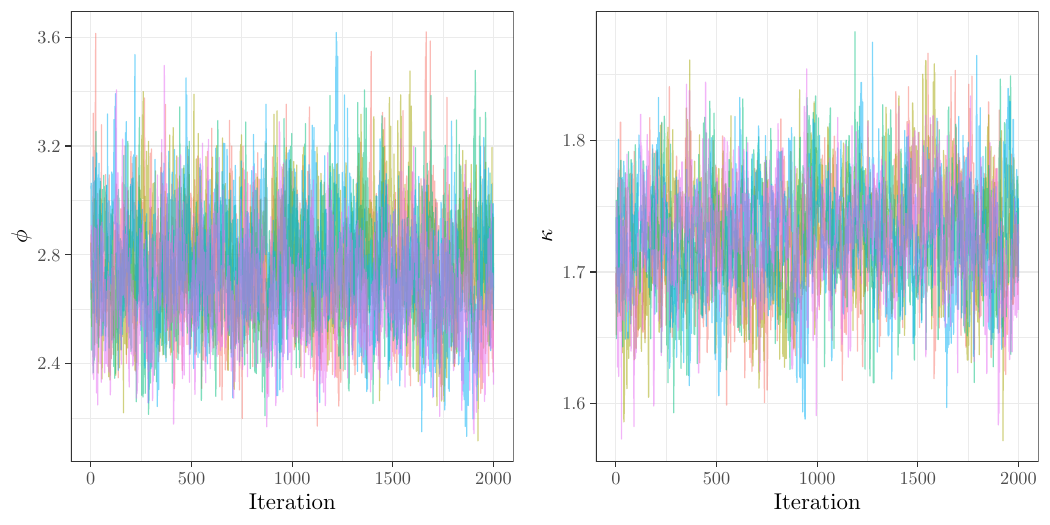}
  \caption{Traceplots for the parameters $\phi$ and $\kappa$ on the MEPS
    dataset for five independent Markov chains.}
  \label{fig:mixing}
\end{figure}

We fit the quasi-power model to this data, with $\log \mu_\theta(x) = r(x)$ and
$V_\eta(\mu) = \mu^\kappa$, using the Laplace approximation to update the tree
topologies as described in Section~\ref{sec:prior-specification} and the
pseudo-likelihood approach to update $(\phi, \kappa)$. For this analysis we ran
five Markov chains in parallel for a total of 2000 iterations per chain, with
the first 1000 iterations of each chain discarded to burn-in. Traceplots for the
parameters $\phi$ and $\kappa$ are displayed in Figure~\ref{fig:mixing}, and we
see that the chains converged rapidly and there were no notable issues with the
mixing. Bearing in mind that our analysis is based on a different subpopulation
than \citet{natarajan2008variance}, our results are quite similar to theirs in
terms of the plausible values of $\kappa$.

\begin{figure}
  \centering
  \includegraphics[width = 1\textwidth]{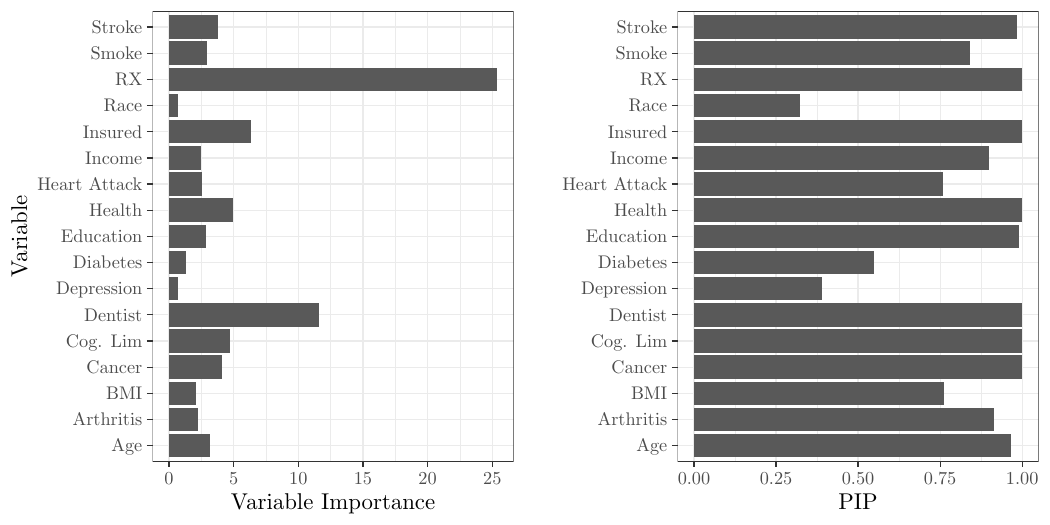}
  \caption{Posterior summarization of variable importances (left) and posterior
    variable inclusion probabilities (right). Variable importances correspond to
    the average number of splits in the ensemble that make use of the listed
    variable.}
  \label{fig:varimp}
\end{figure}

Having verified that the chain is mixing well and that the results for $\kappa$
are broadly in agreement with the existing literature, we move onto summarizing
the posterior distribution. We begin by examining the \emph{posterior inclusion
  probabilities} for each variable \citep{linero2018abayesian} as well as their
variable importances as defined by \citet{chipman2010bart} (defined to be the
posterior mean of the number of times a variable is used in a splitting rule).
Results are displayed in Figure~\ref{fig:varimp}, which shows roughly that the
most important indicators of future medical expenditures are first direct usage
of the healthcare system (number of prescription refills and number of dentist
trips), followed by indicators of poor health (cognitive limitations, cancer,
self-perceived health status etc). Variables that had lower posterior inclusion
probabilities include race, diabetes status, and depression status.

Next, we use the posterior projection approach of \citet{woody2021model} to
create summaries of the function $r(x)$, which can be useful both for
interpreting the model and for identifying interesting subpopulations. For each
sampled $r(x)$ we compute an optimal summary $\widetilde r(x)$ in a given class
of summary functions $\mathcal H$, defined by $\widetilde r = \arg \min_{g \in
  \mathcal H} \sum_i \{r(X_i) - \widetilde r(X_i)\}^2$. We use the 
\emph{summary $R^2$} 
\begin{align*}
  R^2 = 1 - \frac{\sum_i \{r(X_i) - \widetilde r(X_i)\}^2}
  {\sum_i \{r(X_i) - \bar r\}^2} 
  \qquad \text{where} \qquad 
  \bar r = N^{-1} \sum_i r(X_i)
\end{align*}
to measure the quality of the summary, with $R^2$ close to $1$ implying that
$\widetilde r(x)$ is a good approximation of $r(x)$ at the observed values of
the covariates.

In preliminary analysis, we found that the class of generalized additive models
(GAMs) that take $\widetilde r(x) = \sum_{j = 1}^P h_j(x)$ provided very strong
summaries of $r(x)$, with the \emph{partial effect} $h_j(x)$ given by a natural
cubic spline \citep[see][for a review]{wood2006generalized} for the continuous
predictors. The summary $R^2$ for the GAM summaries were found to concentrate
between 95\% and 98\%.

\begin{figure}
  \centering
  \includegraphics[width=.85\textwidth]{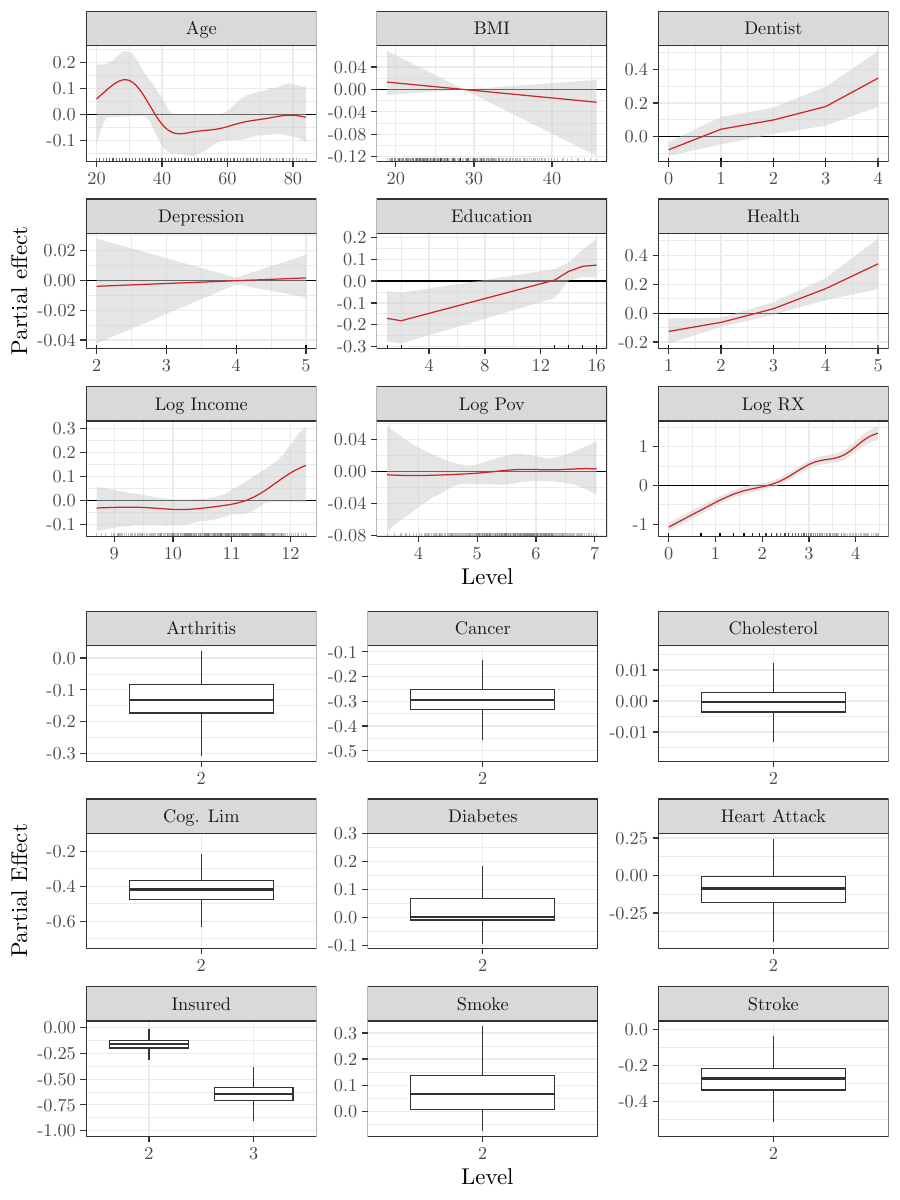}
  \caption{Summary of the quasi-power BART model fit to the MEPS dataset.
    Continuous variables display the posterior mean and confidence bands for
    the associated term of the projection, while categorical variables display
    the posterior median, 50\% credible interval, and 95\% credible interval
    using a boxplot.}
  \label{fig:gam}
\end{figure}

The posterior distributions of the $h_j$'s from the GAM summary of the BART fit
are given in Figure~\ref{fig:gam}; log transformations were used prior to
fitting the GAM summary for the number of perscription refills (RX), income, and
poverty level. For the continuous variables, we see clear increasing trends in
the number of dentist visits, perceived health, and number of prescription
refills. For the categorical variables, the strongest effects are associated to
cancer diagnosis, cognitive limitations, and insurance status.

\section{Discussion}
\label{sec:discussion}

In this paper, we propose to combine quasi-likelihood methods with Bayesian
nonparametric priors on the regression function $\mu_\theta(x)$ to obtain
methods that are both highly practical and robust to model misspecification. The
methods we propose depend only on correctly specifying a mean-variance relation
of the form $\Var_\theta(Y_i \mid X_i = x) = \frac{\phi}{\omega}
V\{\mu_\theta(x)\}$. We develop Bayesian nonparametric versions of the
quasi-binomial, quasi-Poisson, and quasi-gamma models,
and show that all of these approaches tare amenable to simple Bayesian
backfitting algorithms \citep{chipman2010bart} when $\mu_\theta(x)$ is modeled
using a suitable BART prior.

An obvious limitation of our approach is that we are still required to correctly
specify the variance function in order to obtain valid inferences. The variance
model can fail either because (i) the form of $V(\mu)$ is not correctly
specified or (ii) the variance depends on $X_i$ through features other than
$\mu_\theta(X_i)$. We partially alleviate this by allowing for the variance
function itself to be estimated jointly with the dispersion parameter $\phi$
using the Bayesian bootstrap quasi-likelihood (BBQ) update. Both of these issues
can be partially fixed by allowing the dispersion parameter $\phi$ to depend on
covariates $\phi(x)$ using, for example, the \emph{multiplicative decision tree}
model of \citet{pratola2020heteroscedastic}. We defer such extensions to future
work.

This work is part of a larger program of attempting to perform Bayesian
inference while avoiding the specification of a data generating mechanism.
Quasi-likelihood uses only a specification of the mean function (which we model
nonparametrically) and a mean-variance function relationship. One might hope
further to either (i) eliminate the specification of the mean-variance
relationship or (ii) develop a method that is \emph{robust} to misspecification
of the mean-variance relation. A simple approach would be to fully commit to the
Bayesian bootstrap and \emph{define} $\theta$ as the minimum of a
population-level objective function; this works well for parametric problems, as
we have already seen here in our simulation experiments, but it is not clear how
to apply such methods in a non-parametric setting. Other approaches, such as
Bayesian exponentially tilted likelihood (BETEL) or Bayesian method of moments
(BMOM) are limited in scope due to computational challenges
\citep{schennach2005bayesian,chib2018bayesian,yin2009bayesian}. A large part of
the appeal of quasi-likelihood methods for us is that they allow for the use of
simple Bayesian backfitting algorithms to fit BART-based models, which can no
longer be used with these other techniques. Future work will pursue
computationally efficient alternatives to BETEL to further reduce the modeling
assumptions.

\bibliographystyle{apalike}
\bibliography{references.bib}

\appendix

\section{Gibbs Sampler for the Quasi-Poisson Model}

Let $\Leaves(\Tree_t)$ denote the set of leaves of tree $\Tree_t$ and let
$\lambda_\ell$ denote the leaf node parameter associated to leaf $\ell$.
Further, let $\Tree_{-t}$ denote the set of trees excluding tree $t$, let
$\sM_{-t}$ denote the set of leaf node parameters for all trees excluding tree
$t$, and let $\zeta_i = \sum_{m \ne t} g(X_i; \Tree_m, \sM_m)$ denote the sum of
all trees except for the $t$th tree. Following \citet{chipman2010bart}, we use a
two-stage sampling algorithm to update the tree parameters $(\Tree_1, \ldots,
\Tree_T)$. We first compute the integrated likelihood of tree $\Tree_t$ given
$\Tree_{-t}$ and $\sM_{-t}$ as
\begin{align*}
  L(\Tree_t) = \int \prod_{\ell \in \Leaves(\Tree_t)} \prod_{i: X_i \leadsto \ell}
  q_\phi(Y_i \mid \zeta_i + \lambda_\ell) \ d\lambda_\ell,
\end{align*}
where $X_i \leadsto \ell$ denotes the event that $X_i$ is associated to leaf
$\ell$ of tree $t$ and
\begin{align*}
  q_\phi(y \mid \zeta) = \exp\left(-\frac{D\{y , g^{-1}(\zeta)\}}{2\phi}\right),
\end{align*}
is the quasi-likelihood of $y$ given $\zeta$.

For the quasi-Poisson model with $\mu_\theta(x) = \exp\{r(x)\}$ we have $-D(y,
\mu) / 2 = c(y) + y \log \mu - e^{-\mu}$ for some function $c(y)$. The
likelihood associated to $(\sM_t, \Tree_t)$ is therefore proportional to
\begin{align*}
  \prod_{\ell \in \Leaves(\Tree_t)} \prod_{i: X_i \leadsto \ell}
  \exp\left(\frac{Y_i \zeta_i + Y_i \lambda_\ell - e^{\zeta_i} 
  e^{\lambda_\ell}}{\phi}\right)
  \propto \prod_{\ell \in \Leaves(\Tree)} 
  \exp\left(A_\ell \lambda_\ell - B_\ell e^{\lambda_\ell}\right),
\end{align*}
where $A_\ell = \sum_{i: X_i \leadsto \ell} \frac{Y_i}{\phi}$ and $B_\ell =
\sum_{i: X_i \leadsto \ell} \frac{e^{\zeta_i}}{\phi}$. Assuming that
$\lambda_\ell \iid \log \Gam(a, b)$, we can compute the integrated likelihood of
tree $\Tree_t$ from this as
\begin{align*}
  L(\Tree) = \prod_{\ell \in \Leaves(\Tree)} 
  \int_{-\infty}^\infty \frac{b^a}{\Gamma(a)} \exp(a \lambda - b e^\lambda) 
  \times \exp(A_\ell \lambda - B_\ell e^\lambda) \ d\lambda
  = 
  \prod_{\ell \in \Leaves(\Tree)} \frac{b^a}{\Gamma(a)} \times
  \frac{\Gamma(A_\ell + a)}{(B_\ell + b)^{A_\ell + a}}.
\end{align*}
Additionally, we can see that the full-conditional of $\lambda_\ell$ under
this model is given by $\lambda_\ell \sim \log \Gam(A_\ell + a,
b + B_\ell)$. This leads to the following two-step sampling algorithm for
updating $(\Tree_t, \sM_t)$:
\begin{enumerate}
\item Propose a tree $\Tree' \sim Q(\Tree_t \to \Tree')$ using, for example, the
  \texttt{Birth}, \texttt{Depth}, or \texttt{Perturb} moves described by
  \citet{kapelner2016bart} and its appendix.
\item Set $\Tree_t = \Tree'$ if $U \le \frac{L(\Tree') \, \pi(\Tree') \,
    Q(\Tree' \to \Tree_t)}{L(\Tree_t) \, \pi(\Tree_t) \, Q(\Tree_t \to
    \Tree')}$, where $U \sim \Uniform(0,1)$; otherwise, leave $\Tree_t$
  unmodified.
\item For each $\ell \in \Leaves(\Tree_t)$, sample $\lambda_\ell \sim \Gam(a +
  A_\ell, b + B_\ell)$.
\end{enumerate}

\section{Gibbs Sampler for Quasi-Gamma}

For the quasi-gamma model we have
\begin{align*}
  -\frac{D(y, \mu)}{2} = c(y) - \frac{y}{\mu} - \log(\mu).
\end{align*}
We then specify the model $\mu_\theta(x) = e^{-r(x)}$. Defining $\zeta_i$ and
the other tree-based quantities as we did for the quasi-Poisson model
the integrated likelihood of $\Tree_t$ given
$(\Tree_{-t}, \sM_{-t})$ is therefore given by
\begin{align*}
 L(\Tree_t) = &\prod_{\ell \in \Leaves(\Tree_t)} \int \prod_{i : X_i \leadsto \ell}
    \exp\left( \frac{\lambda - Y_i \, e^{\zeta_i} e^\lambda}{\phi} \right)
    \frac{b^a}{\Gamma(a)} \exp(a \lambda - b e^\lambda) \ d\lambda
  \\&\qquad\propto
  \prod_{\ell \in \Leaves(\Tree_t)} \int 
  \exp\left( A_\ell \lambda - B_\ell e^\lambda\right)
  \frac{b^a}{\Gamma(a)} \exp(a\lambda - b e^\lambda) \ d\lambda
  \\&\qquad=
  \prod_{\ell \in \Leaves(\Tree_t)}  \frac{b^a}{\Gamma(a)} \times
  \frac{\Gamma(a + A_\ell)}{(b + B_\ell)^{a + A_\ell}},
\end{align*}
where $A_\ell = N_\ell / \phi$ where $N_\ell$ is the number of observations
assigned to leaf node $\ell$ of tree $\Tree_t$, and $B_\ell = \sum_{i: X_i
  \leadsto \ell} Y_i \, e^{\zeta_i} / \phi$. Similarly, we also have the full
conditional $\lambda_\ell \sim \log \Gam(a + A_\ell, b + B_\ell)$. The algorithm
is now identical to that for the quasi-Poisson model.

\section{Gibbs Sampler for the Quasi-Power Model}

The quasi-power model uses the variance function $V(\mu) = \mu^\kappa$ where
$\kappa$ is estimated based on the data. By routine calculations, the
quasi-likelihood associated to the quasi-power model is given by
\begin{align*}
  q_\phi(y \mid \mu, \kappa)
  = \exp\left\{ \frac{1}{\phi} \left( \frac{y \mu^{1-\kappa}}{1-\kappa}
  - \frac{\mu^{2-\kappa}}{2 - \kappa} \right) \right\},
\end{align*}
for $\kappa \ne 1, 2$, with the special cases of $\kappa = 1,2$ corresponding to
the quasi-Poisson model and the quasi-gamma model, respectively. We consider the
exponential link $\mu_\theta(x) = \exp\{r(x)\}$.

As usual, we let $\zeta_i = \sum_{m \ne t} g(X_i ; \Tree_m, \sM_m)$. The
integrated quasi-likelihood associated $\Tree_t$ given $(\kappa, \phi, \Data,
\Tree_{-t}, \sM_{-t})$ under this model is given by
\begin{align*}
  &\prod_{\ell \in \Tree_t} \int \prod_{i: X_i \leadsto \ell}
    \exp\left\{ \frac{1}{\phi} \left(
    \frac{Y_i e^{\zeta_i (1 - \kappa)} e^{\lambda(1 - \kappa)}}{1 - \kappa} -
    \frac{e^{\zeta_i(2 - \kappa)} e^{\lambda(2 - \kappa)}}{2 - \kappa}
    \right) \right\} \,
    \Normal(\lambda \mid 0, \sigma^2_\lambda)
    \ d\lambda
  \\&= 
  \prod_{\ell \in \Tree_t} \int
  \exp\left\{ 
  \frac{1}{\phi} \left( \frac{A_\ell e^{\lambda(1 - \kappa)}}{1 - \kappa} - 
  \frac{B_\ell e^{\lambda(2 - \kappa)}}{2 - \kappa}\right)
  \right\}
  \Normal(\lambda \mid 0, \sigma^2_\lambda)
  \ d\lambda
\end{align*}
where $A_\ell = \sum_{i: X_i \leadsto \ell} Y_i e^{\zeta_i(1 - \kappa)}$ and
$B_\ell = \sum_{i : X_i \leadsto \ell} e^{\zeta_i (2 - \kappa)}$.

Next, we define $\Lambda_\ell(\lambda) = \frac{1}{\phi} \left( \frac{A_\ell
    e^{\lambda(1-\kappa)}}{1-\kappa} - \frac{B_\ell
    e^{\lambda(2-\kappa)}}{2-\kappa} \right)$. Directly optimizing
$\Lambda_\ell(\lambda)$ shows that it is maximized at $\widehat \lambda_\ell =
\log(A_\ell / B_\ell)$. Proceeding by the Laplace approximation, we
expand about $\widehat \lambda$ to second order to get the approximation
\begin{align}
  \label{eq:treelik}
  \begin{split}
  L(\Tree_t) &\approx\prod_{\ell \in \Tree_t}  e^{\Lambda_\ell(\widehat \lambda_\ell)} \int
    \exp\left\{ -\frac{(\lambda - \widehat \lambda_\ell)^2I(\widehat \lambda)}{2} 
    \right\}
    \Normal(\lambda \mid 0, \sigma^2_\lambda) \ d\lambda
  \\&=
  \prod_{\ell \in \Tree_t} e^{\Lambda_\ell(\widehat \lambda_\ell)}
  \sqrt{\frac{I^{-1}(\widehat \lambda_\ell)}
  {\sigma^2_\lambda + I^{-1}(\widehat \lambda_\ell)}}
  \exp\left\{ -\frac{1}{2}
  \frac{\widehat \lambda^2}
  {(I^{-1}(\widehat\lambda_\ell) + \sigma^2_\lambda)} \right\},
  \end{split}
\end{align}
where $I(\widehat \lambda_\ell) = \left.\frac{\partial^2}{\partial \lambda^2}
  \Lambda_\ell(\lambda)\right|_{\lambda = \widehat\lambda_\ell}$. We use
\eqref{eq:treelik} to compute an approximate acceptance probability for the
Bayesian backfitting algorithm.

After updating $\Tree_t$, we can sample the $\lambda_\ell$'s from their full
conditionals, which are proportional to
\begin{align*}
    \exp\left\{ 
      \frac{1}{\phi} \left( \frac{A_\ell e^{\lambda_\ell(1 - \kappa)}}{1 - \kappa} - 
                            \frac{B_\ell e^{\lambda_\ell(2 - \kappa)}}{2 - \kappa}\right)
    \right\}
    \Normal(\lambda_\ell \mid 0, \sigma^2_\lambda).
\end{align*}
Applying the Laplace approximation again, this full conditional is approximately
proportional to
\begin{align*}
  \exp\left\{ -\frac{(\lambda_\ell - \widehat \lambda_\ell)^2I(\widehat \lambda_\ell)}{2} 
  \right\}
  \Normal(\lambda_\ell \mid 0, \sigma^2_\lambda),
\end{align*}
which is the kernel of a normal distribution with mean and variance given by
\begin{align*}
  m = \frac{I(\widehat \lambda_\ell) \widehat \lambda_\ell}{I(\widehat \lambda_\ell) + \sigma^{-2}_\lambda}
  \qquad \text{and} \qquad
  v = \frac{1}{I(\widehat\lambda_\ell) + \sigma^{-2}_\lambda}.
\end{align*}
This approximation can be further refined by applying, for example, several
rounds of slice sampling \citep{neal2003slice} to the exact full conditional
after sampling $\lambda_\ell$.

As described above, the Gibbs sampler proposed for the quasi-power model is
approximate, and does not hold the exact quasi-posterior invariant (even
ignoring the updates for $\phi$ and $\kappa$). \citet{linero2022generalized}
gives a general recipe for converting Laplace approximations such as these into
proposals that leave the target posterior exactly invariant, and the
modifications can be applied if this is desired. In practice, we find that the
Laplace approximation is adequate on its own for identifying useful tree
topologies, and that the effect of any inaccuracy in the Laplace approximation
for updating $\lambda_\ell$ can be more-or-less eliminated by following up the
sampling with a single slice sampling update. Moreover, because of the large
amount of regularization applied by the prior --- the prior takes
$\sigma^2_\lambda$ very small, especially when many trees are used --- the
Laplace approximation is typically highly accurate without making any
corrections, even if $\Lambda_\ell(\lambda)$ is not close to quadratic.

\section{Gibbs Sampler for the Quasi-Multinomial}

Let $Z_i = n_i \, Y_i$, where $Y_i$ is modeled with a quasi-multinomial model.
Then the quasi-likelihood is given by
\begin{align*}
  Q(\theta) = \prod_i \frac{\exp\{Z_i^\top r(X_i) / \phi\}}
                             {\left[ \sum_k \exp\{r_k(X_i)\} \right]^{n_i/\phi}}.
\end{align*}
We now introduce latent variables $\xi_i \sim \Gam(n_i / \phi, \sum_k
\exp\left\{ r_k(X_i) \right\})$, which leads to the joint likelihood
\begin{align*}
  Q(\theta, \xi)
  &\propto \prod_i
    \exp\left\{ \frac{Z_i^\top r(X_i)}{\phi} - \xi_i \sum_k e^{r_k(X_i)} \right\}
    \xi_i^{n_i / \phi - 1}.
\end{align*}
Defining $\zeta_{ik} = \sum_{m \ne t} [g(X_i; \Tree_t, \sM_t)]_k$, the integrated
likelihood associated to a given tree $\Tree_t$ (conditional on the $\xi$'s as
well as $\Tree_{-t}, \sM_{-t}$, and $\phi$) is then proportional to
\begin{align*}
  &\prod_\ell  \prod_k \int \prod_{i: X_i \leadsto \ell}
  \exp\left\{ \frac{Z_{ik} \, (\zeta_{ik} + \lambda)}{\phi}
    - \xi_i e^{\zeta_{ik}} e^{\lambda} \right\}
  \, \pi(\lambda) \ d\lambda
  \\&\propto \prod_\ell \prod_k \int \exp\left( \frac{A_{\ell k} \lambda}{\phi}
      - B_{\ell k} e^\lambda\right) \, \pi(\lambda) \ d\lambda,
\end{align*}
where $A_{\ell k} = \sum_{i: X_i \leadsto \ell} Z_{ik}$ and $B_{\ell k} =
\sum_{i: X_i \leadsto \ell} \xi_i e^{\zeta_{ik}}$. Assuming the $\log \Gam(a, b)$
prior for the $\lambda_{\ell k}$'s, this gives
\begin{align*}
  L(\Tree_t) \propto
  \prod_\ell \prod_k \frac{b^a}{\Gamma(a)} \times
    \frac{\Gamma(a + A_{\ell k} / \phi)}{(b + B_{\ell k})^{a + A_{\ell k} / \phi}},
\end{align*}
and similarly to previous calculations the full conditional of $\lambda_{\ell
  k}$ is a $\log \Gam(a + A_{\ell k} / \phi, b + B_{\ell k})$ distribution.

\section{Assumptions}
\label{sec:regularity-conditions}

This section provides the regularity conditions that are required to establish
the theoretical results in Section~\ref{sec:theoretical-properties}. First, we
make some assumptions about the form of the quasi-likelihood $Q(\theta)$.

\begin{paragraph}{Condition P}
  The function $r(x) = x^\top \beta$ is related to $\mu_\theta(x)$ the relation
  $g\{\mu_\theta(x)\} = r(x)$ where $g(\mu)$ is a continuous strictly-increasing
  function. Additionally, the function $V_\eta\{g^{-1}(\psi)\}$ is continuous,
  strictly positive, and $D\{Y_i, g^{-1}(\psi)\} = 2 \int_{g^{-1}(\psi)}^y
  \frac{y - t}{V_\eta(t)} \ dt$ exists for all $\psi \in (-\infty, \infty)$.
\end{paragraph}

\vspace{1em}

Next, the following assumption is used to establish the existence of
the stationary distribution.

\begin{paragraph}{Condition A}
  The prior $\pi(\beta)$ is continuous and bounded away from $0$ on its support,
  which is assumed to be a ball $\{\beta : \|\beta\| \le R\}$, and for the PL
  approach $g(\phi \mid \beta)$ is an inverse gamma distribution
  $\text{InvGam}(N/2, \widehat \phi(\beta) N / 2)$ truncated to an interval
  $[a,b]$ with $a > 0$ and $b < \infty$. Additionally, the data $\Data$ is such
  that $Q(\theta)$ is well-defined for all $\theta$ with $\|\beta\| \le R$ and
  $\phi \in [a,b]$.
\end{paragraph}

\vspace{1em}

The following assumption states, essentially, that the model is correctly
specified.

\begin{paragraph}{Condition F}
  The true density of $[Y_i \mid X_i, \omega_i]$, denoted $f_0(y \mid x,
  \omega)$, satisfies the quasi-likelihood moment conditions for some $\theta_0
  = (\phi_0, \beta_0)$ such that $\phi_0 > 0$, $r_0(x) = x^\top \beta_0$ for
  some $\beta_0 \in \Reals^P$, and link function $g\{\mu_0(x)\} = r_0(x)$, with
  the $Y_i$'s independent . Additionally, the prior $\pi(\beta)$ is continuous
  at $\beta_0$ with $\pi(\beta_0) > 0$.
\end{paragraph}

\vspace{1em}

Regularity conditions on the structure of the quasi-likelihood and data
generating mechanism are also required. The following conditions are borrowed
from Appendix A.3 of \citet{agnoletto2023bayesian} and guarantee a pointwise
Bernstein-von Mises result for fixed $\phi$, modified to be slightly simpler to
take advantage of the boundedness assumptions in Condition A  and Condition F.
They hold very generally.

\begin{paragraph}{Condition D}
  Let $X_1, X_2, \ldots$ be the sequence of covariate values and $[Y_i \mid X_i,
  \omega_i]$ be independent and distributed according to Condition F.
  The following conditions hold:
  \begin{itemize}
  \item[D1] We have $\lim_{n \to \infty} \frac{1}{n} \sum_{i=1}^n \omega_i
    \E_{\theta_0}\left(D\{Y_i; g^{-1}(X_i^\top \beta)\} \mid X_i = x \right) =
    D_\infty(\beta)$ for some $D_\infty(\beta)$. Furthermore,
    \begin{align*}
      \lim_{n \to \infty} \sum_{i = 1}^n \frac{\omega_i^2 \Var\left( D\{Y_i; g^{-1}(X_i^\top \beta)\} \right)}{n^2} < \infty.
    \end{align*}
  \item[D2] $D(Y_i; \mu_i)$ is three-times continuously differentiable on the
    support of $\pi(\beta)$, with the third derivative uniformly bounded.
  \item[D3] The matrix $\bX^\top W \bX / N$ converges to a positive definite
    matrix $H(\beta_0)$, where $W = \diag\left( \frac{\omega_i}{g'(\mu_i)^2 \,
        V_\eta(\mu_i)} \right)$ and $\mu_i \equiv g^{-1}(r_0(X_i))$.
  \item[D4] $D_\infty(\beta) > D_\infty(\beta_0)$ if $\beta \ne \beta_0$ for all
    $\beta$ in the support of $\pi(\beta)$.
  \end{itemize}
\end{paragraph}

Finally, we assume that the moment estimator is sufficiently regular in the
sense that $\widehat \phi(\beta_n) \to \phi_0$ provided that $\beta_n \to
\beta_0$ sufficiently fast.

\begin{paragraph}{Condition G}
  Let $\widehat \phi(\beta)$ denote the moment estimator $\frac{1}{n} \sum_i
  \frac{\omega_i \{Y_i - g^{-1}(X_i^\top\beta)\}^2}{V_\eta(g^{-1}(X_i^\top\beta))}$,
  and let $\widehat \beta = \arg \min_\beta \sum_i \omega_i D\{Y_i;
  g^{-1}(X_i^\top \beta)\}$ be the maximum quasi-likelihood estimator of
  $\beta$. Then, with probability $1$, $\widehat \phi(\widehat \beta + h / \sqrt
  n) \to \phi_0$ for all $h$.
\end{paragraph}

\vspace{1em}

The above conditions also imply correct Frequentist properties of the
quasi-posterior, as under standard regularity conditions (and assuming that the
mean and variance relations are correctly specified) it is well-known that
$\widehat \beta \stackrel{{}_{\bullet}}{\sim} N(\beta_0, \phi N^{-1}
H(\beta_0)^{-1})$. 

\section{Theoretical Results}

\subsection{Existence of the Stationary Distribution}

We note that Condition A, which truncates the support of $\pi(\theta)$ to a
compact set, is overly strong and serves to eliminate technical issues in
verifying that the chain is irreducible, aperiodic, and positive recurrent,
which are standard sufficient conditions for the existence of a stationary
distribution \citep[see][Chapter 10]{meyn2012markov}.

\begin{proof}

  Let $\sS = \{(\beta, \phi) : \|\beta\| \le R, \phi \in [a,b]\}$. Note that the
  Markov transition function (MTF) of the two-step Gibbs sampler is given by
  \begin{align*}
    R\{(\beta, \phi) \to (\beta', \phi')\}
    =
    \pit(\beta' \mid \phi) \, g(\phi' \mid \beta').
  \end{align*}
  From Condition A and Condition P it follows that $\pit(\beta' \mid \phi) \,
  g(\phi' \mid \beta')$ is bounded away from $0$ on the set $\sS \times \sS$. We
  now apply Theorem 1.3.1 (ii) of \citet{meyn2012markov} to conclude the existence
  of a unique stationary distribution:
  \begin{itemize}
  \item The set $\sS$ is small because the MTF is bounded away from $0$ on $\sS$,
    as this implies that the MTF is minorized by a uniform measure on $\sS$.
  \item The hitting time for $\sS$ is trivially $1$ regardless of how the chain is
    initialized.
  \end{itemize}
  Hence Theorem 1.3.1 (i) and (ii) imply the result.

  Next, we characterize the stationary distribution $\pit(\beta, \phi)$. If we
  sample $(\beta, \phi) \sim \pit(\beta,\phi)$ and then sample the transition
  $R\{(\beta, \phi) \to (\beta', \phi')\}$ then, by definition, $(\beta', \phi')
  \sim \pit(\beta, \phi)$, but from the structure of the Gibbs sampler it is clear
  that the conditional associated to $\phi'$ is $g(\phi' \mid \beta')$; hence,
  $\pit(\beta, \phi) = \pit(\beta) \, g(\phi \mid \beta)$ for some $\pit(\beta)$.

  To characterize $\pit(\beta)$, we use the definition of the stationary
  distribution to obtain
  \begin{align*}
    \pit(\beta') \, g(\phi' \mid \beta')
    &=\int \pit(\beta, \phi) \pit(\beta' \mid \phi)
      \, g(\phi' \mid \beta') \ d\beta d\phi
    \\&= g(\phi' \mid \beta') \int \pit(\beta, \phi) 
    \pit(\beta' \mid \phi) \ d\beta  \ d\phi
    \\&= g(\phi' \mid \beta') \int m(\phi) \, \pit(\beta' \mid \phi) \ d\phi,
  \end{align*}
  where $m(\phi)$ is the marginal distribution of $\phi$ for the stationary
  distribution. Hence $\pit(\beta) = \int \pit(\beta \mid \phi) \, m(\phi) \
  d\phi$.

  To obtain the result $m(\phi) = \int g(\phi \mid \beta) \, \pit(\beta) \
  d\beta$, we repeat the same logic with the roles of $\phi$ and $\beta$
  reversed, and analyze the \emph{dual chain} that reverses the order of
  updating $\beta$ and $\phi$; note that, because our sampler is not a genuine
  Gibbs sampler, the dual chain does not have the same stationary distribution
  as the original chain, however they do have the same marginal distributions,
  which is sufficient.
\end{proof}

\subsection{Bernstein-von Mises Theorem}

We begin by proving a \emph{uniform} Bernstein-von Mises theorem, establishing
uniform convergence of the posterior $\pit(\beta \mid \phi)$. Throughout, we use
the symbol $A \lesssim B$ to mean that $A \le C B$ for some positive constant $C
> 0$ that is allowed to depend only on $(P, \beta_0, a, b, H, \pi)$.

\begin{lemma}
  \label{lem:bvm}
  Suppose that Condition P, Condition A, Condition F, and D1--D4 hold. Then
  \begin{align*}
    \sup_{\phi \in [a,b]} \int |q_n^\phi - N^\phi| \ dh \to 0 \qquad \text{almost surely}
  \end{align*}
  where $q_n^\phi(h)$ is the quasi-posterior distribution of $\sqrt n(\beta -
  \widehat \beta_n)$, $\widehat\beta_n$ is the maximum quasi-likelihood
  estimator of $\beta$, and $N^\phi$ is the density of a normal distribution
  with mean $0$ and variance $\phi \, H^{-1}(\beta_0)$.
\end{lemma}

\begin{proof}
  First, Theorem 3 of \citet{agnoletto2023bayesian} implies the non-uniform
  version of this statement: for all $\phi \in [a,b]$, we have $\int |q_n^\phi -
  N^\phi| \ dh \to 0$ almost surely (the statement of the result is
  in-probability, but as a corollary of results of \citet{miller2021asymptotic}
  the techniques used prove almost-sure convergence). Now, fix $\delta > 0$ and
  let $a = \lambda_1 < \lambda_2 < \cdots \lambda_K = b$ be an evenly-space
  $\delta$-cover of $[b^{-1}, a^{-1}]$, i.e., $\lambda_{k + 1} - \lambda_k \le
  \delta$; let $\mathcal K$ denote this $\delta$-cover.

  For fixed $\phi$, let $\lambda_k = \lambda$ be such that $\lambda_k \le \phi^{-1}
  \le \lambda_{k + 1}$. Then
  \begin{align*}
    \int |q^\phi_n - N^\phi| \ dh
    \le \underbrace{\int |q_n^\phi - q_n^{\lambda^{-1}}| \ dh}_{(i)}
      + \underbrace{\int |q_n^{\lambda^{-1}} - N^{\lambda^{-1}}| \ dh}_{(ii)}
      + \underbrace{\int |N^{\lambda^{-1}} - N^\phi| \ dh}_{(iii)}.
  \end{align*}
  We now bound (i) -- (iii) uniformly in $\phi$.

  \vspace{1em}

  \noindent\textbf{Second term:} From the pointwise version of the result
  we have that $\max_{\lambda \in \mathcal K} \int |q_n^{\lambda^{-1}} -
  N^{\lambda^{-1}}| \to 0$ as $n \to \infty$ almost surely; hence (ii) converges
  almost surely $0$ in probability uniformly in $\phi$.
  
  \vspace{1em}

  \noindent \textbf{Third term:} By Pinsker's inequality, the third term is
  bounded by a multiple of
  \begin{align*}
    \sqrt{\log \lambda^{-1} / \phi 
    + \phi / \lambda^{-1} - 1}.
  \end{align*}
  Letting $\Delta = \frac{\lambda^{-1}}{\phi} - 1$ and bounding $\log(1 +
  \Delta) \le \Delta$ we have the bound
  \begin{align*}
    \sqrt{\log(1 + \Delta) - \Delta/(1+\Delta)}
    \le \frac{\Delta}{\sqrt{1 + \Delta}}.
  \end{align*}
  Now, $\Delta = \frac{\phi^{-1} - \lambda}{\lambda} \le b \delta$, while
  $1 + \Delta = \phi^{-1} / \lambda \ge a / b$ so that we can bound this last
  term by $\delta \times \sqrt{b / a}$. That is,
  \begin{align*}
    \sup_\phi \int |N^{\lambda^{-1}} - N^\phi| \ dh
    \lesssim \delta.
  \end{align*}
  
  \vspace{1em}

  \noindent \textbf{First term:} We first bound the total variation distance
  with the Hellinger distance:
  \begin{align}
    \label{eq:hell}
    \int |q_n^{\lambda^{-1}} - q_n^\phi| \ dh
    \lesssim 1 - \int \sqrt{q_n^\phi \, q_n^{\lambda^{-1}}} \ dh
    = 1 - \sqrt{\frac{[Z_n^{(\rho + \lambda)/2}]^2}{Z_n^\rho \, Z_n^{\lambda}}}
  \end{align}
  where $\rho \equiv \phi^{-1}$, 
  \begin{align*}
    Z_n^\rho &= \int \exp\left\{ \log Q_\phi(\widehat \beta + h / \sqrt n) - \log Q_\phi(\widehat \beta) \right\} \, \pi(\widehat \beta + h / \sqrt n) \ dh
    \\&=
    \int \exp\left\{ \rho [\log Q_1(\widehat \beta + h / \sqrt n) - \log Q_1(\widehat \beta)]\right\} \pi(\widehat \beta + h / \sqrt n) \ dh
  \end{align*}
  and the equality in \eqref{eq:hell} (after routine calculation) comes from the
  fact that
  \begin{align*}
    q_n^\phi(h) &\propto Q_\phi(\widehat \beta + h / \sqrt n)
      \, \pi(\widehat \beta + h / \sqrt n) \qquad \text{where} 
    \\Q_\phi(\beta) &= \exp\left\{ -\frac{\sum_{i=1}^n \omega_i D\{Y_i; g^{-1}(X_i^\top\beta)\}}{2\phi} \right\}
  \end{align*}
  is the quasi-likelihood. Next, we observe that for all $\lambda_k \le \rho \le
  \lambda_{k + 1}$ we have $(Z_n^{\lambda_k})^{\rho / \lambda_k} \le Z_n^\rho
  \le (Z_n^{\lambda_{k+1}})^{\rho / \lambda_{k+1}}$ by Jensen's inequality.
  Accordingly, the Hellinger distance is bounded above by
  \begin{align*}
    1 - \frac{( Z_n^{\lambda_{k+1}})^{(\lambda_k + \rho) / 2 {\lambda_{k+1}}}}{(Z_n^{\lambda_k})^{(\lambda_k + \rho)/ 2 \lambda_k}}
    = 1 - \left( \frac{Z_n^{\lambda_{k+1}}}{Z_n^{\lambda_{k}}} \right)^{(\lambda + \rho) / (2\lambda_{k+1})}
    \times \left( Z_n^{\lambda_{k}} \right)^{\frac{\lambda + \rho}{2} \times \frac{\lambda_{k} - \lambda_{k+1}}{\lambda_{k}}}.
  \end{align*}
  Next, we note that
  \begin{align*}
    \frac{\lambda_k + \rho}{2\lambda_{k+1}}  \log \left( \frac{Z_n^{\lambda_{k+1}}}{Z_n{^\lambda_{k}}} \right)
    &\ge - \left| \log \left( \frac{Z_n^{\lambda_{k+1}}}{Z_n^{\lambda_{k}}} \right)  \right| \equiv A_{nk}
    \qquad \text{and} \qquad
    \\
    -\frac{\lambda_k + \rho}{2} \times \frac{\lambda_{k+1} - \lambda_{k}}{\lambda_{k}} \log Z_n^{\lambda_{k}}
    &\ge - \frac{\delta \lambda_{k+1}}{\lambda_{k}} |\log Z_n^{\lambda_{k}}| \equiv B_{nk}.
  \end{align*}
  As these bounds do not depend on $\phi$, we have
  \begin{align*}
    \sup_\phi \int |q_n^\phi - q_n^{\lambda^{-1}}|
    \lesssim \max_k 1 - e^{A_{nk} + B_{nk}}.
  \end{align*}
  By Theorem 4 of \citet{miller2021asymptotic}, we have the pointwise Laplace
  approximation
  \begin{align*}
    \log Z_n^\rho \to
    \log \pi(\beta_0) + \frac{P\log(\rho)}{2} + \frac{1}{2} \log |2\pi H(\beta_0)|
    \qquad \text{almost surely}.
  \end{align*}
  Hence, $A_{nk} \to \frac{P}{2} \log(\lambda_{k} / \lambda_{k+1})$ and $B_{nk}
  \to \frac{-\delta \lambda_{k+1}}{\lambda_k} \left| \log \pi(\beta_0) + \frac{P
      \log \lambda_k}{2} + \frac{1}{2} \log |2\pi H(\beta_0)|\right|$. Letting
  \begin{align*}
    A = -\frac{P}{2} \frac{\delta}{a}
    \qquad \text{and} \qquad
    B = -\delta \left\{|\log \pi(\beta_0)| + \frac{P (|\log a| + |\log b|)}{2} +
    \frac{1}{2} |2\pi H(\beta_0)|\right\}
  \end{align*}
  we have $\lim A_{nk} \ge A$ and $\lim B_{nk} \ge B$ almost surely.
  Therefore
  \begin{align*}
    \limsup \sup_\phi \int |q_n^\phi - q_n^{\lambda^{-1}}|
    \lesssim 1 - e^{A+ B} = 1 - e^{-\delta C} \le \delta C
  \end{align*}
  where $C = -(A + B) / \delta > 0$ does not depend on $\delta$.

  \noindent \textbf{Putting Everything Together:} From the above arguments, we
  have
  \begin{align*}
    \limsup \sup_\phi \int |q^\phi_n - N^\phi| \ dh \lesssim
    \delta + 0 + \delta.
  \end{align*}
  Because this holds for all $\delta > 0$, we must have
  \begin{align*}
    \lim_{n \to \infty} \sup_\phi \int |q^\phi_n - N^\phi| \ dh = 0
  \end{align*}
  almost surely.
\end{proof}

We use this uniform Bernstein-von Mises theorem to establish the result for the
quasi-posterior defined in
Section~\ref{sec:theoretical-properties}. 

\begin{proof}
  We start by showing that the Bernstein-von Mises result applies to a slightly
  different distribution, which is sampled from via the following algorithm:
  \begin{enumerate}
  \item Initialize $\sigma$ with any value in $[a, b]$.
  \item Sample $\beta' \sim \pit(\beta' \mid \sigma)$.
  \item Sample $\phi \sim g(\phi \mid \beta')$.
  \item Sample $\beta \sim \pit(\beta \mid \phi)$.
  \end{enumerate}
  Let $m_\sigma$ denote the marginal distribution of $\phi$ under this algorithm
  and $\pit_*$ the marginal distribution of $\beta$ under this algorithm. By the
  triangle inequality we have
  \begin{align*}
    \int |\pit_* - N(\cdot \mid \beta_0, \phi_0 N^{-1} \, H(\beta_0)^{-1})|
    &\le 
    \int |\pit_* - \pit(\cdot \mid \phi_0)|  
    \\&\qquad +\int |\pit(\cdot \mid \phi_0) -
      N(\cdot \mid \beta_0, \phi_0 N^{-1} \, H(\beta_0)^{-1})|.
  \end{align*}
  By Lemma~\ref{lem:bvm}, the second term converges to $0$ almost surely, so it
  suffices to show that the first term $A = \int |\pit_*(\beta) - \pit(\beta \mid
  \phi_0)| \ d\beta$ also converges to $0$ almost surely. Next, we rewrite
  $\pit_*(\beta)$ as
  \begin{align*}
    \int \pit(\beta \mid \phi) \, m_\sigma(\phi) \ d\phi
    \qquad \text{where} \qquad
    m_\sigma(\phi) = \int g(\phi \mid \beta') \, \pit(\beta' \mid \sigma).
  \end{align*}
  By the triangle inequality we have
  \begin{align*}
    A &\le \int |\pit(\beta \mid \phi) - \pit(\beta \mid \phi_0)| \,
      m_\sigma(\phi) \ d\phi \ d\beta
    \\&\le \int |\pi(\beta \mid \phi) -
      N(\beta \mid \widehat \beta, \phi \Sigma)|
    \, m_\sigma(\phi) + 
    |N(\beta \mid \widehat \beta, \phi \Sigma) -
    N(\beta \mid \widehat \beta, \phi_0 \Sigma)| \ m_\sigma(\phi)
    \\&\qquad + 
    |\pit(\beta \mid \phi_0) -
    N(\beta \mid \widehat \beta, \phi_0 \Sigma)| m_\sigma(\phi) \ d\phi \ d\beta
  \end{align*}
  where $\Sigma = (N H(\beta_0))^{-1}$. By another application of
  Lemma~\ref{lem:bvm}, it suffices to show that $\int |N^\phi - N^{\phi_0}| \,
  m_\sigma(\phi) \ d\phi \ d\beta \to 0$ almost surely, where $N^\phi$ is
  defined as in proof of Lemma~\ref{lem:bvm}. To show this, we show that
  $m_\sigma(\phi) \to \delta_{\phi_0}(\phi)$ in distribution (almost surely).
  Letting $\psi(\phi)$ denote an arbitrary bounded and continuous function, we
  have
  \begin{align*}
    \int \psi(\phi) \, m_\sigma(\phi) \ d\phi
    &= \iint \psi(\phi) \, g(\phi \mid \widehat \beta + h / \sqrt n)
      q_n^\sigma(h) \ d\phi \ dh
    \\&= \int \E_g\{\psi(\phi) \mid \widehat \beta + h / \sqrt n\}
      q_n^\sigma(h) \ dh
  \end{align*}
  where $q_n^\sigma$ is as defined in the proof of Lemma~\ref{lem:bvm}. Because
  $\psi(\phi)$ is bounded, $q_n^\sigma \to N^\sigma$ pointwise by the Laplace
  approximation, and $\int q_n^\sigma \ dh \to \int N^\sigma \ dh$ we can apply
  the generalized dominated convergence theorem to conclude that
  \begin{align*}
    \lim_{n \to \infty} \int \psi(\phi) \, m(\phi) \ d\phi
    = \int \lim_{n \to \infty} \E_g\{\psi(\phi) \mid \widehat \beta + h / \sqrt n\} q^\sigma(h) \ d\phi
  \end{align*}
  almost surely. By the truncated gamma specification of $g(\phi \mid \beta)$ we
  have
  \begin{align*}
    [\phi \mid \widehat \beta + h / \sqrt n]
    \sim \text{InvGam} \left\{
    \frac{n}{2}, \frac{n \widehat \phi(\widehat \beta + h / \sqrt n)}{2}
    \right\}
    \text{ truncated to } [a, b].
  \end{align*}
  By Condition G (i.e., $\widehat \phi(\widehat \beta + h / \sqrt n) \to \phi_0$
  for all $h$ almost surely) and the fact that $\phi_0 \in [a,b]$ by assumption,
  it is straight-forward to show that this distribution converges to a
  $\delta_{\phi_0}$ distribution (almost surely), and hence $\lim_{n \to \infty}
  \E_g\{\psi(\phi) \mid \widehat \beta + h / \sqrt n\} = \psi(\phi_0)$ almost
  surely. Therefore
  \begin{align*}
    \lim_{n \to \infty} \int \psi(\phi) \, m(\phi) \ d\phi
    = \int \psi(\phi_0) N^\sigma(h) \ dh = \psi(\phi_0),
  \end{align*}
  so that $m_\sigma(\phi) \to \delta_{\phi_0}(\phi)$ in distribution (almost
  surely). It follows that
  \begin{align*}
    \int |N^\phi - N^{\phi_0}| m_\sigma(\phi) \ dh \ d\phi
    = \int \Xi(\phi) \, m_\sigma(\phi) \ d\phi \to \Xi(\phi_0) = 0
  \end{align*}
  where $\Xi(\phi) = \int |N^\phi - N^{\phi_0}| \ dh$.
  
  Having established the result for $\pit_*$, we now establish the result for
  $\pit$. By an identical argument to the one above, it suffices to show that
  $m(\phi) \to \delta_{\phi_0}(\phi)$ in distribution (almost surely), which is
  also part of the conclusion of the theorem. We use a strategy similar to the
  one used to prove Lemma~\ref{lem:bvm}. Let $b^{-1} = \lambda_0 < \lambda_1 <
  \cdots < \lambda_K = a^{-1}$ denote a $\delta$-net of $[b^{-1}, a^{-1}]$ and
  for each $\phi$ let $\sigma_\phi = \lambda_k^{-1}$ such that $\lambda_k \le
  \phi^{-1} \le \lambda_{k+1}$. Using the integral representation of $m(\phi)$
  we have (for an arbitrary bounded and continuous function $\psi(\phi)$)
  \begin{align*}
    \int \psi(\phi) m(\phi) \ d\phi
    = \int \psi(\phi) \, g(\phi \mid \beta) \,
      \pit(\beta \mid \rho) \ m(\rho) \ d\rho \ d\beta \ d\phi
  \end{align*}
  Adding and subtracting $\pit(\beta \mid \sigma_\rho)$ we get
  \begin{align*}
    &\int \psi(\phi) g(\phi \mid \beta)
      \{\pit(\beta \mid \rho) - \pit(\beta \mid \sigma_\rho)\} \,
      m(\rho) \ d\rho \ d\beta \ d\phi
    \\&\qquad + \int \psi(\phi) g(\phi \mid \beta) \,
      \pit(\beta \mid \sigma_\rho) \, m(\rho) \ d\rho \ d\beta \ d\phi
  \end{align*}
  We conclude by showing that the first term can be made arbitrarily small while
  the second term above tends to $\psi(\phi_0)$. For the first term we apply the
  triangle inequality, boundedness of $\psi(\phi_0)$, and change of variables to
  get
  \begin{align*}
    &\left| \int \psi(\phi) g(\phi \mid \beta)
      \{\pit(\beta \mid \rho) - \pit(\beta \mid \sigma_\rho)\} \, m(\rho)
      \ d\rho \ d\beta \ d\phi \right|
    \\&\lesssim \int |q_n^\rho - q_n^{\sigma_\rho}| \, m(\rho) \ dh \ d\rho.
  \end{align*}
  By the argument given under the heading \textbf{First Term} in the proof of
  Lemma~\ref{lem:bvm}, we know that $\limsup \sup_\rho \int |q_n^\rho -
  q_n^{\sigma_\rho}| \ dh \lesssim \delta$ almost surely so that
  \begin{align*}
    \limsup \left| \int \psi(\phi) g(\phi \mid \beta) \{\pit(\beta \mid \rho) - \pit(\beta \mid \sigma_\rho)\} \, m(\rho) \ d\rho \ d\beta \ d\phi \right|
    \lesssim \delta.
  \end{align*}
  For the second term, let $m\{(a, b)\}$ denote the mass $m(\phi)$ assigned to
  the interval $(a,b)$ and observe that
  \begin{align*}
    &\int \psi(\phi) g(\phi \mid \beta) \pit(\beta \mid \sigma_\rho)
      \, m(\rho) \ d\rho \ d\beta \ d\phi
    \\&= \sum_k m\{(\lambda_k, \lambda_{k+1})\} \int \psi(\phi)
    \, g(\phi \mid \beta) \pit(\beta \mid \sigma_k) \ d\beta \ d\phi.
    \\&= \sum_k m\{(\lambda_k, \lambda_{k+1})\} \int \psi(\phi)
    \, m_{\sigma_k}(\phi) \ d\phi
    \\&\to \sum_k m\{(\lambda_k, \lambda_{k+1})\} \psi(\phi_0) = \psi(\phi_0),
  \end{align*}
  with the last line following from the fact that $m_{\sigma_k}(\phi) \to
  \delta_{\phi_0}(\phi)$ in distribution (almost surely).

  Summarizing, we have
  \begin{align*}
    &\limsup \left| \int \psi(\phi) \, m(\phi) - \psi(\phi_0) \right|
    \\&\qquad\le
    \limsup \left|\int \psi(\phi) g(\phi \mid \beta)
    \{\pit(\beta \mid \rho) - \pit(\beta \mid \sigma_\rho)\} \, m(\rho) \
    d\rho \ d\beta \ d\phi  \right|
    \\&\qquad \qquad + \left|\int \psi(\phi) g(\phi \mid \beta)
    \, \pit(\beta \mid \sigma_\rho) \, m(\rho) \ d\rho \ d\beta \ d\phi
      - \psi(\phi_0) \right|
    \\&\qquad\lesssim \delta + 0,
  \end{align*}
  for all $\delta > 0$, so that the limit exists and is $0$.
\end{proof}

\end{document}

%% file: Figure/qmn_results.tex
\begin{tabular}{lrrrrr}
  \toprule
  Method & RMSE & Width & Coverage  \\ 
  \midrule
  BART   & 0.08 & 0.28  & 0.92      \\ 
  QMN    & 0.05 & 0.20  & 0.95      \\ 
  \bottomrule
\end{tabular}